\documentclass[journal]{IEEEtran}

\usepackage[cmex10]{amsmath}  
\usepackage{amssymb}
\usepackage{graphicx,color}   

\usepackage{ifxetex}
\ifxetex
\usepackage{xunicode}
\else
\usepackage[utf8]{inputenc}
\usepackage[T1]{fontenc}
\fi

\usepackage{cite}  
\usepackage[caption=false,font=footnotesize]{subfig} 

\newtheorem{thm}{Theorem}
\newtheorem{defn}{Definition}
\newtheorem{lem}{Lemma}
\newtheorem{remark}{Remark}

\newtheorem{algorithm}{Algorithm}

\newcommand{\argmax}{\operatornamewithlimits{argmax}}

\newcommand{\arglmax}{\operatornamewithlimits{arglmax}}

\newcounter{mytempeqncnt}

\begin{document}

\title{On Capacity and Optimal Scheduling for the Half-Duplex Multiple-Relay Channel}

\author{Lawrence Ong, Mehul Motani, Sarah J. Johnson\\
\thanks{The material in this paper was presented in part at
the 46th Annual Allerton Conference on Communication, Control, and Computing, Monticello, IL, September 2008,
the International Symposium on Information Theory and its Applications, Auckland, New Zealand, December 2008, and
the IEEE International Symposium on Information Theory, Seoul, South Korea, June 2009.}
\thanks{Lawrence Ong and Sarah J. Johnson are with School of Electrical Engineering and Computer Science, The University of Newcastle, Callaghan, NSW 2308, Australia (e-mail: lawrence.ong@cantab.net; sarah.johnson@newcastle.edu.au).}
\thanks{Mehul Motani is with National University of Singapore, Electrical \& Computer Engineering, 10 Kent Ridge Crescent, Singapore 119260 (e-mail: motani@nus.edu.sg).}
}

\maketitle

\begin{abstract}
We study the half-duplex multiple-relay channel (HD-MRC) where every node can either transmit or listen but cannot do both at the same time.
We obtain a capacity upper bound based on a max-flow min-cut argument and achievable transmission rates based on the decode-forward (DF) coding strategy, for both the discrete memoryless HD-MRC and the phase-fading HD-MRC. We discover that both the upper bound and the achievable rates are functions of the transmit/listen state (a description of which nodes transmit and which receive). More precisely, they are functions of the time fraction of the different states, which we term a schedule.
We formulate the optimal scheduling problem to find an optimal schedule that maximizes the DF rate. The optimal scheduling problem turns out to be a maximin optimization, for which we propose an algorithmic solution. We demonstrate our approach on a four-node multiple-relay channel, obtaining closed-form solutions in certain scenarios. Furthermore, we show that for the received signal-to-noise ratio degraded phase-fading HD-MRC, the optimal scheduling problem can be simplified to a max optimization.
\end{abstract}

\begin{IEEEkeywords}
Decode-forward (DF), half duplex, multiple-relay channel (MRC), scheduling, phase fading.
\end{IEEEkeywords}

\section{Introduction}

In this paper, we investigate the $D$-node  half-duplex multiple-relay channel (HD-MRC). The $D$-node multiple-relay channel (MRC) models a network where a source transmits data to a destination with the help of $(D-2)$ relays, which themselves have no data to send. In half-duplex channels, a node can either listen or transmit but cannot do both at the same time. Analyses of half-duplex channels are appealing from a practical point of view as they model most transceivers available in the market today.

Though much has been done to understand the full-duplex multiple-relay channel (FD-MRC)~\cite{xiekumar03,kramergastpar04,ongmotaniit08a} in which all nodes can transmit and listen at the same time, its half-duplex counterpart receives less attention for it is often thought of as a special case of the full-duplex channel, where a node simply stops listening when it transmits and vice versa. However, we will show in this paper that half-duplex channels identify issues, e.g., scheduling, that do not appear in the full-duplex counterpart. We will also show that the coordination of the transmitting and the listening modes among the relays plays an important role in determining the bottleneck of the half-duplex channel.

For the half-duplex networks, the \emph{transmit/listen state} is defined to describe which nodes are transmitting (or listening) at any time (by default, nodes not transmitting are considered to be listening). In the $D$-node HD-MRC, where the source always transmits and the destination always listens, the total number of transmit/listen states is $2^{D-2}$, capturing which of the $(D-2)$ relays are transmitting.

In this paper, we first obtain an upper bound to the capacity of the HD-MRC by specializing the cut-set bound for the half-duplex multiterminal networks~\cite{khojastepoursabharwal03}. We then use the idea of the decode-forward (DF) coding strategy for the FD-MRC~\cite{xiekumar03} to derive a lower bound to the capacity (i.e., achievable rate) of the HD-MRC.
We will see that the capacity upper bound and the achievable DF rates depend on the time fractions of the transmit/listen states, but do not depend on the sequence or the order of the states. This means we do not need to coordinate the transmit-listen sequence among the nodes to maximize the achievable DF rate or to find the capacity upper bound. We can view the combination of time fractions of different transmit/listen states as a \emph{schedule}. We then formulate the optimal scheduling problem to find an optimal schedule that maximizes the DF rate for the HD-MRC.

The optimal DF scheduling problem turns out to be a maximin problem, which is not easily solved in general.
We propose a greedy algorithm to solve the general maximin problem. We then show how the algorithm can be simplified when applied to the optimal DF scheduling problem. Using the simplified algorithm, we obtain closed-form solutions for certain channel topologies. Furthermore, we show that for the received signal-to-noise ratio (rSNR) degraded phase-fading HD-MRC, the optimal scheduling problem can be further simplified to a max optimization.
We compute DF rates of HD-MRCs and compare them to that of the corresponding FD-MRCs. We highlight some differences between the half-duplex and the full-duplex channels.

The specific contributions of this paper are as follows:
\begin{enumerate}
\item We obtain an upper bound to the capacity of the discrete memoryless HD-MRC and the phase-fading HD-MRC, by specializing the cut-set bound for the general half-duplex multiterminal network~\cite{khojastepoursabharwal03}.
\item We derive achievable rates for the discrete memoryless HD-MRC and the phase-fading HD-MRC. The coding scheme is constructed based on decode-forward (a coding strategy for the FD-MRC~\cite{xiekumar03}) and the idea of cycling through all transmit/listen states in each codeword (an extension of that used for the half-duplex single-relay channel~\cite{hostmadsen02}). Note that this result on achievable rates was independently derived by us, in an earlier
conference version of this work~\cite{ongwangmotani08allerton,ongwangmotani08isita}, and by Rost and Fettweis~\cite{rostfettweis08a,rostfettweis08b}.
\item We formulate the optimal scheduling problem for the phase-fading HD-MRC to find the optimal schedule that gives the highest DF rate, and propose an algorithm to solve the scheduling problem.
\item We show that for the rSNR-degraded phase-fading HD-MRC, the optimal scheduling problem can be simplified to a max optimization.
\item We compare the DF rates of the HD-MRC to those of the FD-MRC. The bottleneck nodes in the FD-MRC depend on the channel topology. On the other hand, the bottleneck of the phase-fading HD-MRC is always constrained by the first $B$ relays, where $1 \leq B \leq D-2$, or by all the relays and the destination. For the case of rSNR-degraded phase-fading HD-MRC, the bottleneck is always constrained by all the relays and the destination, which is independent of the topology.
\end{enumerate}

\subsection{Related Work}

The Gaussian half-duplex single-relay channel (HD-SRC), i.e., the MRC with $D=3$, was first studied by H{\o}st-Madsen~\cite{hostmadsen02} who derived a capacity upper bound based on cut-set arguments and a capacity lower bound based on the partial-decode-forward relaying strategy where the relay decodes part of the source message. A year later Khojastepour et al.~\cite{khojastepoursabharwalaazhang03,khojastepoursabharwalaazhang03b} showed that partial-decode-forward achieves the capacity of the degraded (Gaussian and discrete memoryless) HD-SRC, a result that mirrors the full-duplex counterpart where decode-forward achieves the capacity of the degraded single-relay channel~\cite{covergamal79}. In the aforementioned work, the transmit/listen state for each channel use is fixed and known a priori to all nodes. In 2004, Kramer~\cite{kramer04} showed that larger DF rates could potentially be achieved using random states because the choice of state could be used by the source and the relay to transmit information.

Other strategies have been proposed for the HD-SRC, e.g. (i) compress-forward where the relay quantizes its received signals and forwards them to the destination~\cite{hostmadsenzhang05}, and (ii) one where the relay decodes every second source message~\cite{kramer07ima} (i.e., the relay toggles between the transmit and listen modes for alternate blocks of message transmission of the source). The fading HD-SRC has been studied by H{\o}st-Madsen and Zhang~\cite{hostmadsenzhang05} and Rankov and Wittneben~\cite{rankovwittneben07}. A hybrid scheme using full-duplex and half-duplex relaying was proposed for the single-relay channel by Yamamoto et al.~\cite{yamamotohaneda10}.

For the half-duplex \emph{multiple}-relay channel (HD-MRC) (where $D \geq 3$), Khojastepour et al.\ considered the \emph{cascaded} MRC where the source, the relays, and the destination form a ``chain'', and each node only receives the transmissions from one node (i.e., the \emph{upstream} node in the chain). The capacity of the cascaded HD-MRC was derived for the Gaussian case~\cite{khojastepoursabharwalaazhang03b} as well as the discrete memoryless case~\cite{khojastepoursabharwalaazhang03}. Another class of the HD-MRC where the relays only receive signals from the source has been studied~\cite{gastparvetterli05,delcosoibars07,xuesandhu07}.

The general HD-MRC  where each relay (in its listening mode) can receive the transmissions of all other nodes was first studied by Kramer et al.~\cite{kramergastpar04}, where they mentioned that DF for the FD-MRC can be extended to the HD-MRC. However, explicit expressions for achievable rates were not derived.
In 2008, in the conference version of this paper~\cite{ongwangmotani08allerton,ongwangmotani08isita}, we studied this class of general HD-MRCs, and derived an upper bound to the capacity based on cut-set arguments and a lower bound to the capacity based on DF. In the same papers we introduced scheduling for the HD-MRC and algorithms to optimize the schedule to achieve the maximum DF rate. Independently in the same year, Rost and Fettweis derived capacity lower bounds for (i) the discrete memoryless HD-MRC\cite{rostfettweis08a} using partial-decode-forward and compress-forward strategies, (ii) the Gaussian HD-MRC~\cite{rostfettweis08b} using partial-decode-forward and amplify-forward strategies. Their results on achievable rates include ours as a special case.\footnote{In contrast to this paper, optimization of states, or scheduling, was not considered by Rost and Fettweis \cite{rostfettweis08a,rostfettweis08b}.} The scaling behavior of the HD-MRC, i.e., when the number of relays grow to infinity, was studied by Gastpar and Vetterli~\cite{gastparvetterli05}.

In half-duplex networks, the question of scheduling arises, i.e., how the relays coordinate their transmit/listen modes. Scheduling for the HD-MRC with two relays (i.e., $D=4$) and where there is no direct source-destination link has been recently studied by Rezaei et al.~\cite{rezaeigharan10}. It was shown that when the signal-to-noise ratio (SNR) is low and under certain conditions, the optimal schedule spans over two states: (i) both relays listen, and (ii) both relays transmit. For the high SNR scenario, the optimal schedule is such that at any time only one relay listens and the other transmits. In this paper, we study scheduling for the Gaussian HD-MRC for any number of relays assuming that DF coding strategy is used. We propose an algorithm to find the optimal schedule for DF.

In the literature, scheduling for the relay network is also studied under other context, e.g., (i) resource allocation: allocating time fractions for relays to forward data to different destinations~\cite{hammerstromkuhn04,weogotlin05,shizhang08}, and (ii) interference avoidance:  avoiding simultaneous transmissions among nearby relays (graph coloring)~\cite{leechen07,guomaguo08}.

\section{Channel Models}\label{sec:hd-mrc}

\subsection{Discrete Memoryless HD-MRC}
The FD-MRC was introduced by Gupta and Kumar~\cite{guptakumar03}, and Xie and Kumar~\cite{xiekumar03}. In a $D$-node MRC with nodes $\{1, 2, \dotsc, D\} \triangleq \mathcal{D}$, node $1$ is the source, node $D$ the destination, and nodes $2$ to $(D-1)$ the relays. A message $W$ is generated at node $1$ and is to be sent to node $D$.  In this paper, we consider the HD-MRC where a node can only transmit ($T$) or listen ($L$) at any point in time. We assume that the source is always transmitting, and the destination is always listening.

We use a transmit/listen state vector $\boldsymbol{s}$ to capture the half-duplex scenarios among the relays. For the $D$-node HD-MRC where the source always transmits and the destination always listens, the state can be expressed as $\boldsymbol{s}=(s^{(1)}, \dotsc, s^{(i)}, \dotsc, s^{(D-2)}) \in \{L,T\}^{D-2}$, where the $i$-th element of $\boldsymbol{s}$ is $T$ if node $(i+1)$ transmits, and is $L$ otherwise, i.e., if node $(i+1)$ listens.
We will use the terms transmit/listen state and state interchangeably in the rest of this paper.
In this paper, we only consider the case where the state is fixed and known a priori to all nodes.

We define the set of all relays as $\mathcal{R} = \{2, 3, \dotsc, D-1\}$. In state $\boldsymbol{s}$, we further define $\mathcal{L}(\boldsymbol{s})$ as the set of relays that are listening and $\mathcal{T}(\boldsymbol{s})$ as the set of relays that are transmitting, i.e., $\mathcal{L}(\boldsymbol{s}) \triangleq \{ i \in \mathcal{R} : s^{(i-1)} = L\}$ and $\mathcal{T}(\boldsymbol{s}) \triangleq \{ i \in \mathcal{R} : s^{(i-1)} = T\}$.  Note that $\mathcal{L}(\boldsymbol{s}) \cap \mathcal{T}(\boldsymbol{s}) = \emptyset$ and $\mathcal{L}(\boldsymbol{s}) \cup \mathcal{T}(\boldsymbol{s}) = \mathcal{R}$.

Throughout this paper, for each random variable denoted by an upper letter $X$, its realization is denoted by the corresponding lower letter $x$. We use the calligraphic font to denote sets, and bold letters to denote vectors. For example, $\boldsymbol{X} = (X[1], X[2], \dotsc, X[n])$, and for $\mathcal{B} = \{b_1, b_2, \dotsc, b_{|\mathcal{B}|}\}$, $X_{\mathcal{B}}=\{X_{b_1}, X_{b_2}, \dotsc, X_{b_{|\mathcal{B}|}}\}$.

We denote the input from user $i$ to the channel by $X_i \in \mathcal{X}_i$, and the channel output received by user $i$ by $Y_i \in \mathcal{Y}_i$. When a node $i$ is listening, i.e., not transmitting, we assign a dummy symbol as its input to the channel, $\tilde{x}_i \in \mathcal{X}_i$, where $\tilde{x}_i$ is fixed and is known a priori to all nodes.  Similarly, when node $i$ is transmitting, we assign $\tilde{y}_i \in \mathcal{Y}_i$ to be its (dummy) received channel output, where $\tilde{y}_i$ is also fixed and known a priori to all nodes.  For example, in the Gaussian half-duplex channel where channel inputs and outputs are taken over the set of real numbers, we can set $\tilde{x}_i = 0$ if node $i$ is listening or $\tilde{y}_i =0$ if node $i$ is transmitting.

The $D$-node HD-MRC in state $\boldsymbol{s}=(L, s^{(2)}, \dotsc, s^{(D-3)},T)$ is depicted in Fig.~\ref{fig:hd-mrc}. In this state, node 2 is listening and node $(D-1)$ is transmitting.

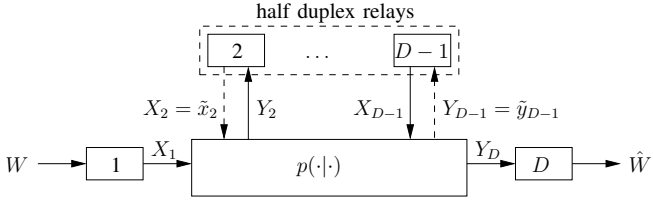
\begin{figure}
\centering
\resizebox{0.95\linewidth}{!}{ 
\begin{picture}(0,0)%
\includegraphics{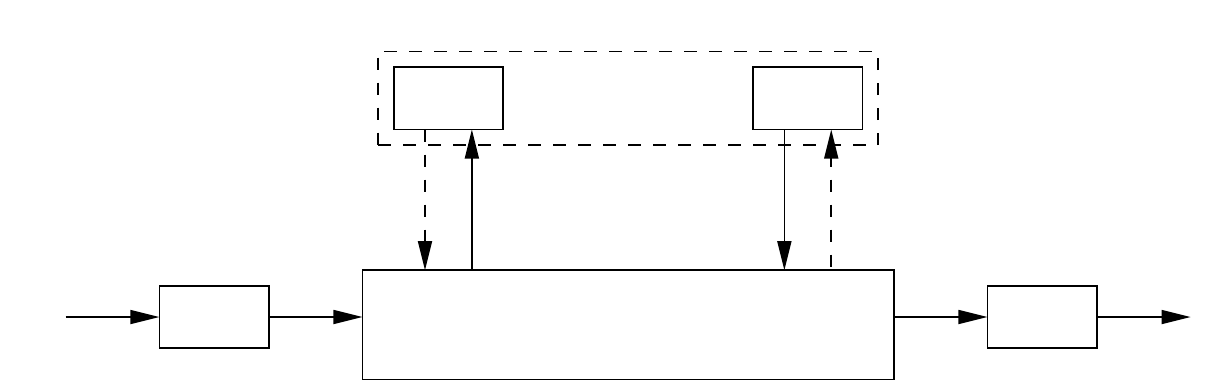}%
\end{picture}%
\setlength{\unitlength}{3947sp}%
\begingroup\makeatletter\ifx\SetFigFont\undefined%
\gdef\SetFigFont#1#2#3#4#5{%
  \fontsize{#1}{#2pt}%
  \fontfamily{#3}\fontseries{#4}\fontshape{#5}%
  \selectfont}%
\fi\endgroup%
\begin{picture}(5805,1821)(-389,-1198)
\put(601,-961){\makebox(0,0)[lb]{\smash{{\SetFigFont{12}{14.4}{\familydefault}{\mddefault}{\updefault}{\color[rgb]{0,0,0}1}%
}}}}
\put(4501,-961){\makebox(0,0)[lb]{\smash{{\SetFigFont{12}{14.4}{\familydefault}{\mddefault}{\updefault}{\color[rgb]{0,0,0}$D$}%
}}}}
\put(1726, 89){\makebox(0,0)[lb]{\smash{{\SetFigFont{12}{14.4}{\familydefault}{\mddefault}{\updefault}{\color[rgb]{0,0,0}2}%
}}}}
\put(3226, 89){\makebox(0,0)[lb]{\smash{{\SetFigFont{12}{14.4}{\familydefault}{\mddefault}{\updefault}{\color[rgb]{0,0,0}$D-1$}%
}}}}
\put(-374,-961){\makebox(0,0)[lb]{\smash{{\SetFigFont{12}{14.4}{\familydefault}{\mddefault}{\updefault}{\color[rgb]{0,0,0}$W$}%
}}}}
\put(5401,-961){\makebox(0,0)[lb]{\smash{{\SetFigFont{12}{14.4}{\familydefault}{\mddefault}{\updefault}{\color[rgb]{0,0,0}$\hat{W}$}%
}}}}
\put(976,-811){\makebox(0,0)[lb]{\smash{{\SetFigFont{12}{14.4}{\familydefault}{\mddefault}{\updefault}{\color[rgb]{0,0,0}$X_1$}%
}}}}
\put(3976,-811){\makebox(0,0)[lb]{\smash{{\SetFigFont{12}{14.4}{\familydefault}{\mddefault}{\updefault}{\color[rgb]{0,0,0}$Y_D$}%
}}}}
\put(2851,-436){\makebox(0,0)[lb]{\smash{{\SetFigFont{12}{14.4}{\familydefault}{\mddefault}{\updefault}{\color[rgb]{0,0,0}$X_{D-1}$}%
}}}}
\put(3676,-436){\makebox(0,0)[lb]{\smash{{\SetFigFont{12}{14.4}{\familydefault}{\mddefault}{\updefault}{\color[rgb]{0,0,0}$Y_{D-1}=\tilde{y}_{D-1}$}%
}}}}
\put(1951,-436){\makebox(0,0)[lb]{\smash{{\SetFigFont{12}{14.4}{\familydefault}{\mddefault}{\updefault}{\color[rgb]{0,0,0}$Y_2$}%
}}}}
\put(901,-436){\makebox(0,0)[lb]{\smash{{\SetFigFont{12}{14.4}{\familydefault}{\mddefault}{\updefault}{\color[rgb]{0,0,0}$X_2=\tilde{x}_2$}%
}}}}
\put(2401, 89){\makebox(0,0)[lb]{\smash{{\SetFigFont{12}{14.4}{\familydefault}{\mddefault}{\updefault}{\color[rgb]{0,0,0}$\dotsc$}%
}}}}
\put(2326,-961){\makebox(0,0)[lb]{\smash{{\SetFigFont{12}{14.4}{\familydefault}{\mddefault}{\updefault}{\color[rgb]{0,0,0}$p(\cdot|\cdot)$}%
}}}}
\put(1951,464){\makebox(0,0)[lb]{\smash{{\SetFigFont{12}{14.4}{\familydefault}{\mddefault}{\updefault}{\color[rgb]{0,0,0}half duplex relays}%
}}}}
\end{picture}%
}
\caption{The HD-MRC in state $\boldsymbol{s}=(L,\dotsc,T)$ where node 2 listens and node $(D-1)$ transmits}
\label{fig:hd-mrc}
\end{figure}

Since the HD-MRC is an instance of the FD-MRC in a certain transmit/listen state, the channel distribution of the HD-MRC follows the underlying FD-MRC. Let the channel distribution for a $D$-node FD-MRC be
\begin{equation}
p^*(y_2, y_3, \dotsc, y_D| x_1, x_2, \dotsc, x_{D-1})\nonumber
\end{equation}
on $\mathcal{Y}_2 \times \mathcal{Y}_3 \times \dotsm \times \mathcal{Y}_D$, for each $(x_1, x_2, \dotsc, x_{D-1}) \in \mathcal{X}_1 \times \mathcal{X}_2 \times \dotsm \times \mathcal{X}_{D-1}$.
The channel distribution for the corresponding $D$-node HD-MRC in state $\boldsymbol{s}$ is given by
\newpage
\begin{multline}
p(y_2, y_3, \dotsc, y_D| x_1, x_2, \dotsc, x_{D-1}, \boldsymbol{s}) \\ = p \big(y_{\mathcal{L}(\boldsymbol{s})}, y_D | x_1, x_2, \dotsc, x_{D-1}, y_{\mathcal{T}(\boldsymbol{s})} \big) \prod_{i \in \mathcal{T}(\boldsymbol{s})} \delta(y_i = \tilde{y}_i),\label{eq:model-hd-mrc}
\end{multline}
where
\begin{multline}
p \big(y_{\mathcal{L}(\boldsymbol{s})}, y_D | x_1, x_2, \dotsc, x_{D-1}, y_{\mathcal{T}(\boldsymbol{s})} \big)\\ = \frac{p^*(y_2, y_3, \dotsc, y_D| x_1, x_2, \dotsc, x_{D-1})}{\sum_{y_{\mathcal{L}(\boldsymbol{s})}, y_D}p^*(y_2, y_3, \dotsc, y_D| x_1, x_2, \dotsc, x_{D-1})}.
\end{multline}
$\delta(A)$ is the indicator function that is 1 if the event $A$ is true, and is 0 otherwise. The expression in \eqref{eq:model-hd-mrc} forces that with probability equals one, node $i$'s received signal will be $\tilde{y}_i$ if it is transmitting.  In this paper, we only consider memoryless and time-invariant channels.

\subsection{Block Codes and Achievable Rates}

In the MRC, the information source at node 1 emits a random messages $W \in \{ 1, ..., 2^{nR} \} \triangleq \mathcal{W}$, which is to be sent to the destination, node $D$. We consider block coding with $n$ uses of the channel. An $( 2^{nR}, n )$ code for the $D$-node HD-MRC comprises the following:
\begin{enumerate}
\item (source) An encoding function at node 1, $f_{1} : \mathcal{W} \rightarrow \mathcal{X}_1^n$, such that $\boldsymbol{X}_1=f_1(W)$, which maps a source message to a codeword of length $n$. Here $\boldsymbol{X}_1 = (X_1[1], X_1[2], \dotsc, X_1[n])$ where $X_1[t]$ is node 1's input to the channel at time $t$.
\item (relays) $n$ encoding functions at each node $i \in \{2,\dotsc,D-1\}$, $f_{i,t}: \mathcal{Y}_i^{t-1} \rightarrow \mathcal{X}_i$, for all $t \in \{1,\dotsc,n\}$, such that $X_i[t] = f_{i,t}(Y_i[1], Y_i[2], \dotsc, Y_i[t-1])$, which map past received signals to the next transmit signal.
\item (relays) The restriction that at the $t$-th channel use in state $\boldsymbol{s}$, the transmit messages of the relays that are listening are restricted to $X_i[t] = \tilde{x}_i$, $\forall i \in \mathcal{L}(\boldsymbol{s})$.
\item (destination) A decoding function at node $D$, $g_{D}: \mathcal{Y}_D^{n} \rightarrow \mathcal{W}$, such that $\hat{W}=g_{D}(\boldsymbol{Y}_D)$,
which maps $n$ received signals to a source message estimate.
\end{enumerate}

On the assumption that the source message $W$ is uniformly distributed over $\mathcal{W}$, the average error probability is defined
as $P_e = \Pr \{\hat{W} \neq  W \}$. The rate $R$ is achievable if, for any $\epsilon > 0$, there is at least one $(2^{nR},n)$ code such that $P_e < \epsilon$. The capacity $C$ is defined as the supremum of the achievable rates.

\subsection{Phase-Fading HD-MRC}\label{sec:hd-gaussian-mrc}

Now, we define the phase-fading HD-MRC. We set $\tilde{x}_i=0$ for nodes that are listening, and $\tilde{y}_i=0$ for nodes that are transmitting. In state $\boldsymbol{s}$, the received signal at node $k \in \{2,\dotsc, D\}$ is given by
\begin{equation}
Y_k =
\begin{cases}
\displaystyle \sum_{i \in \{1\}\cup \mathcal{T}(\boldsymbol{s})} \sqrt{\lambda_{i,k}} e^{j\theta(i,k)} X_i + Z_k, \;\; \text{if }k \in \mathcal{L}(\boldsymbol{s}) \cup \{D\}\\
\tilde{y}_k =0, \quad\quad\quad\quad\quad\quad\quad\;\; \text{otherwise, i.e., if } k \in \mathcal{T}(\boldsymbol{s}),
\end{cases}
\end{equation}
where each $X_i$ is a complex random variable, $Z_k$ is the receiver noise at node $k$ and is an i.i.d. circularly symmetric complex Gaussian random variable with variance $E[Z_k Z_k^\dagger] = N_k$, where $Z_k^\dagger$ is the complex-conjugate of $Z_k$, $\lambda_{i,k}$ captures the path loss from node $i$ to node $k$ and is $\kappa d_{i,k}^{-\eta}$ for $d_{i,k} \geq 1$, and is $\kappa$ otherwise, $d_{i,k} \geq 0$ is the distance between nodes $i$ and $k$, $\eta \geq 2$ is the attenuation exponent (with $\eta = 2$ for free space transmission), $\kappa$ is a positive constant, $e^{j\theta(i,k)}$ is the phase-fading random variable, $\theta(i,k)$ is uniformly distributed over $[0,2\pi)$, and $\theta(i,k)$ are jointly independent of each other for all channel uses, for all $i$, and for all $k$.

We assume that all nodes know $\kappa$ and $d_{i,k}$ for all node pairs. We also assume that node $k$ only knows $\theta(i,k)$ for all $i$, and does not know any $\theta(i,l)$ for $l \neq k$. Hence, the transmitted signals of node $i$ cannot be chosen as a function of $\theta(i,k)$ for any $k$.

In this paper, we consider the following individual-node \emph{per-symbol} transmitted power constraint. Setting the half-duplex constraints $\tilde{x}_i = 0$ for node $i$ in the listening mode, we get
\begin{equation}
E[X_iX_i^\dagger] \leq
\begin{cases}
P_i, & \text{if } i \in \mathcal{T}(\boldsymbol{s}) \cup \{1\}\\
0, & \text{otherwise, i.e., if } i \in \mathcal{L}(\boldsymbol{s}).
\end{cases}
\end{equation}

\section{Upper Bound to the Capacity}\label{sec:hd-capacity-upper-bound}

\subsection{Capacity Upper Bound for the Discrete Memoryless HD-MRC}
An upper bound to the capacity of the HD-MRC is given in the following theorem.
\begin{thm}\label{thm:hd-mrc-upper-bound}

Consider the discrete memoryless $D$-node HD-MRC. The capacity $C$ is upper bounded by
\begin{align}
C \leq &\sup_{p(x_1, x_2, \dotsc, x_{D-1},\boldsymbol{s})} \min_{\mathcal{Q} \subseteq \mathcal{R}} \nonumber  \\
& \sum_{\boldsymbol{s} \in \{L,T\}^{D-2}} p(\boldsymbol{s})
I\Big( X_1,X_{\mathcal{Q}\cap\mathcal{T}(\boldsymbol{s})}; Y_{\mathcal{Q}^c\cap\mathcal{L}(\boldsymbol{s})},Y_D \Big|
 \nonumber \\
 &  X_{\mathcal{Q}^c\cap\mathcal{T}(\boldsymbol{s})},  X_{\mathcal{L}(\boldsymbol{s})} = \tilde{x}_{\mathcal{L}(\boldsymbol{s})}, Y_{\mathcal{Q}^c\cap\mathcal{T}(\boldsymbol{s})}=\tilde{y}_{\mathcal{Q}^c\cap\mathcal{T}(\boldsymbol{s})}, \boldsymbol{S}=\boldsymbol{s} \Big), \label{eq:hd-hd-dm-cut-set-rate}
\end{align}
where $p(x_1, x_2, \dotsc, x_{D-1},\boldsymbol{s})=p(\boldsymbol{s})p(x_1,x_{\mathcal{T}(\boldsymbol{s})}|x_{\mathcal{L}(\boldsymbol{s})}, \boldsymbol{s})$ $\prod_{k\in\mathcal{L}(\boldsymbol{s})} \delta(x_k=\tilde{x}_k)$, $\mathcal{Q}^c = \mathcal{R} \setminus \mathcal{Q}$, $\mathcal{R}$ is the set of all relays, $\mathcal{T}(\boldsymbol{s})$ is the set of all relays that are transmitting in state $\boldsymbol{s}$, and $\mathcal{L}(\boldsymbol{s})$ is the set of all relays that are listening in state $\boldsymbol{s}$.
\end{thm}

The aforementioned upper bound is obtained from the cut-set bound for the half-duplex multiple-source multiple-destination network~\cite{khojastepoursabharwal03}, which is a generalization of the cut-set bound for the full-duplex multiple-source multiple-destination network~\cite[Theorem 15.10.1]{coverthomas06}. Consider all \emph{cuts} of the network separating $\{1\} \cup \mathcal{Q}$ from $\mathcal{Q}^c \cup \{D\}$, for all possible selections of relays $\mathcal{Q} \subseteq \mathcal{R}$. We see that the source is always on one side of the cut with the destination always on the other side. For simplicity, we call the set  $\{1\} \cup \mathcal{Q}$ the \emph{left side} of the cut and $\mathcal{Q}^c \cup \{D\}$ the \emph{right side} of the cut. An upper bound to the rate can be obtained by considering the maximum achievable rate across the cut, from the left to the right, assuming that all nodes on each side of the cut can fully cooperate. More specifically, consider any cut, the rate of message flow from the source to the destination (across the cut) is upper bounded by the mutual information between the inputs from all transmitting nodes on the left side of the cut and the outputs to all listening nodes on the right side, conditioned on the inputs and outputs of all transmitting nodes on the right side, and the inputs from all listening nodes (which is known a priori to all nodes). Hence, any achievable rate must be upper bounded by that for all the cuts (i.e., those with the source on the left side of the cut and the destination on the right side of the cut). The proof of Theorem~\ref{thm:hd-mrc-upper-bound} can be found in Appendix~\ref{appendix:ub}.

\subsection{Capacity Upper Bound for the Phase-Fading HD-MRC}
For the phase-fading HD-MRC, we set the inputs from transmitting nodes to the channel to be independent Gaussian. Note that channel inputs that take advantage of coherent combining are not feasible because the phase changes randomly for each channel use and is unknown to the transmitter. So, we get the following capacity upper bound.
\begin{thm}\label{thm:hd-gaussian-mrc-upper-bound}
Consider a $D$-node phase-fading HD-MRC. The capacity $C$ is upper bounded by
\begin{equation}
C \leq \max_{p(\boldsymbol{s})} \min_{\mathcal{Q} \subseteq \mathcal{R}} r_{(\mathcal{Q})} \triangleq R_\text{ub},
\end{equation}
where
\begin{align}
r_{(\mathcal{Q})} = &\sum_{\boldsymbol{s} \in \{L,T\}^{D-2}} \Bigg[ p(\boldsymbol{s})\nonumber  \\
& \quad \Gamma \Bigg(  \sum_{k \in (\mathcal{Q}^c\cap\mathcal{L}(\boldsymbol{s})) \cup \{D\}} \frac{\sum_{i \in \{1\} \cup (\mathcal{Q}\cap\mathcal{T}(\boldsymbol{s}))} \lambda_{i,k}P_i}{N_k}  \Bigg) \Bigg],
\end{align}
where $\Gamma(x) = \log_2(1+x)$.
\end{thm}

\begin{IEEEproof}[Proof for Theorem~\ref{thm:hd-gaussian-mrc-upper-bound}]
Theorem~\ref{thm:hd-gaussian-mrc-upper-bound} follows directly from  Theorem~\ref{thm:hd-mrc-upper-bound} by using independent Gaussian inputs for all nodes.
\end{IEEEproof}

\begin{remark}\label{remark:gaussian-optimal}
See \cite[Lemma 1]{hostmadsenzhang05}, \cite[Theorems 6 \& 7]{kramergastpar04} for the optimality of independent Gaussian inputs in phase-fading channels.
\end{remark}

\section{Lower Bound to the Capacity}\label{sec:hd-achievability}

\subsection{Achievable DF Rates for the Discrete Memoryless HD-MRC}

Now, we present achievable rates for the discrete memoryless HD-MRC using DF. For a chosen probability mass function of the states $p(\boldsymbol{s})$ and a fixed input distribution for each state $p(x_1, x_2, \dotsc, x_{D-1}|\boldsymbol{s})$, we define the \emph{reception rate} of node $i \in \{2,3,\dotsc,D\}$ as follows:
\begin{align} \label{eq:reception-rate}
r_i = & \sum_{\substack{\boldsymbol{s}\in\{L,T\}^{D-2}\\\text{s.t. } i \in \mathcal{L}(\boldsymbol{s})\cup \{D\}}}\Big[ p(\boldsymbol{s}) I \Big(X_1, X_{\{2, \dotsc, i-1\} \cap \mathcal{T}(\boldsymbol{s})}; Y_i \Big| \nonumber \\
& \quad X_{\{i,\dotsc,D-1\} \cap \mathcal{T}(\boldsymbol{s})}, X_{\mathcal{L}(\boldsymbol{s})}=\tilde{x}_{\mathcal{L}(\boldsymbol{s})},
\boldsymbol{S} = \boldsymbol{s}\Big) \Big].
\end{align}
The reception rate of node $i$ is the rate at which node $i$ can decode the source message using DF. The following rate is achievable on the discrete memoryless HD-MRC.

\begin{thm}\label{thm:df-hd-k-rc}
Consider the $D$-node HD-MRC. Rates up to the following value are achievable:
\begin{equation}
R_\text{DF} = \max_{p(\boldsymbol{s})} \max_{p(x_1, x_2, \dotsc, x_{D-1}|\boldsymbol{s})} \min_{i \in \{2,3,\dotsc,D\}} r_i,\label{eq:df-df-k-rc}
\end{equation}
for any $p(x_1, x_2, \dotsc, x_{D-1}|\boldsymbol{s}) = \prod_{k\in\mathcal{L}(\boldsymbol{s})} \delta(x_k=\tilde{x}_k) p(x_1,x_{\mathcal{T}(\boldsymbol{s})}|x_{\mathcal{L}(\boldsymbol{s})},\boldsymbol{s})$, where $r_i$ is defined in \eqref{eq:reception-rate}.
\end{thm}

\begin{remark}
The first maximization in \eqref{eq:df-df-k-rc} is taken over all possible schedules, i.e., the probability mass functions of the states. The second maximization is taken over all possible input distributions for each state. The minimization is taken over the reception rates of all the relays and the destination---using DF, each relay and the destination must fully decode the source message. Note that the maximum DF rate only depends on the fractions of the states (i.e., the schedule), and does not depend on the sequence of the states.
\end{remark}

The aforementioned achievable rates are derived where the sequence in which the source message is decoded at the nodes is $1, 2, \dotsc, D$. This means node $i$ decodes from all nodes with indices smaller than $i$ (we term these nodes the {\em upstream} nodes), and transmits to all nodes with indices larger than $i$ (we term these nodes the {\em downstream } nodes). 
In general, different decoding sequences give different DF rates. Hence, one can optimize \eqref{eq:df-df-k-rc} over all possible sequences with the constraint that node 1 being the first node and node $D$ being the last node in the sequence~\cite{ongmotani10tcom}.

\indent\indent {\it Sketch of Proof of Theorem~\ref{thm:df-hd-k-rc}: } Here, we briefly summarize the proof of Theorem~\ref{thm:df-hd-k-rc}. Details are given in Appendix~\ref{appendix:df}. We use the idea from DF for the FD-MRC by Xie and Kumar~\cite[Section  IV.2]{xiekumar03}, which uses \emph{regular} encoding and \emph{sliding-window} (simultaneous) decoding technique. Block coding is used, and in each block of $n$ channel uses, a node's transmit signals are functions of its newly decoded message and of its previously decoded messages that its downstream nodes are forwarding. Decoding of a message, as the term window-decoding implies, happens across a few blocks. A node decodes a message from all upstream nodes using a few blocks of received signals. For the construction of the codewords, we use the idea from the HD-SRC, where each codeword is split into two parts (in the first part, the relay listens; in the second part, the relay transmits)~\cite{hostmadsen02}. Extending this idea to the MRC (assuming that all the nodes are synchronized), in each block, we cycle through all different states of the nodes' transmit/listen modes. In each codeword (a block of $n$ channel uses), if a node only listens in $n'$ ($n' \leq n$) channel uses, the rate at which it can decode the message will be reduced (compared to the full-duplex coding) accordingly. Furthermore, among the $n'$ received signals, not all upstream nodes are transmitting. This means that a node will only get signals from transmitting nodes in these $n'$ channel uses, and the rate at which it can decode the message can be determined accordingly. This is the reason why the order of the states does not affect the achievable rates, but the schedule does. 

\begin{remark}
Note that the achievable rate in Theorem~\ref{thm:df-hd-k-rc} is different from simply taking the time-sharing average of the DF rates for the FD-MRC~\cite[Theorem 3.1]{xiekumar03} over all possible transmit/listen states. Consider the single-relay channel (i.e., $D=3$) as an example. The DF rate for the full-duplex case is $R_\text{DF}^\text{f} = \min \{I(X_1;Y_2|X_2), I(X_1,X_2;Y_3)\}$. For its half-duplex counterpart, there are two possible states: $S=L$ and $S=T$. When the relay listens [$S=L$], we have $X_2 = \tilde{x}_2$, meaning that $I(X_1;Y_2|X_2) = I(X_1;Y_2|X_2=\tilde{x}_2)$ and $I(X_1,X_2;Y_3)=I(X_1;Y_3)$. So $R_{\text{DF}[S=L]}^\text{f}= \min \{I(X_1;Y_2|X_2=\tilde{x}_2), I(X_1;Y_3)\}$.  When the relay transmits [$S=T$], we have $Y_2 = \tilde{y}_2$, meaning that $I(X_1;Y_2|X_2) = 0$. So $R_{\text{DF}[S=T]}^\text{f}=0$. Denoting $\alpha =p(S=L)$ and taking the time-sharing average of these two DF rates, we obtain the following rate on the HD-MRC:
\begin{multline}
\alpha R_{\text{DF}[S=L]}^\text{f} + (1-\alpha) R_{\text{DF}[S=T]}^\text{f} \\ =  \min \{\alpha I(X_1;Y_2|X_2=\tilde{x}_2),\alpha I(X_1;Y_3)\}.
\end{multline}
The aforementioned rate can only be lower than the DF rate obtained in Theorem~\ref{thm:df-hd-k-rc} (using the same chosen $\alpha$), i.e.,
\begin{align}
R_\text{DF} =  \min \Big\{ &\alpha I(X_1;Y_2|X_2=\tilde{x}_2),  \nonumber \\ &\alpha I(X_1;Y_3)  + (1-\alpha) I(X_1,X_2;Y_3) \Big\}.
\end{align}
\end{remark}

\begin{remark}
In this paper, we only consider the ``full'' decode-forward strategy where all relays must fully decode the source message, and that imposes a constraint on each relay. Higher rates could possibly be obtained if we only use a subset of the relays to perform decode-forward (see, e.g., \cite{delcosoibars07,ongmotaniit08a}). Of course, omitting some relays also means that we cannot use these relays to send data to the downstream relays and the destination. Other papers have considered (i) using the relays that do not decode the source message to perform compress-forward operations~\cite{rostfettweis08a,rostfettweis08b}, and (ii) having the relays decode only parts of the source message~\cite{delcosoibars07,rostfettweis08a,rostfettweis08b}.
\end{remark}

\subsection{Achievable DF Rates for the Phase Fading HD-MRC}

For phase-fading channels, we have the following achievable rates:

\begin{thm}\label{thm:df-gaussian-hd-krc}
Consider the phase-fading $D$-node HD-MRC. Rates up to the following value are achievable:
\begin{equation}
R_\text{DF} = \max_{p(\boldsymbol{s})} \min_{i \in \{2,3,\dotsc,D\}} r_i, \label{eq:hd-gaussian-df-rate}
\end{equation}
where
\begin{equation}
r_i = \sum_{\substack{\boldsymbol{s} \in\{L,T\}^{D-2}\\\text{s.t. } i \in \mathcal{L}(\boldsymbol{s})\cup \{D\}}} \left[ p(\boldsymbol{s}) \Gamma \Bigg( \frac{\sum_{k \in \{1\} \cup ( \{2, \dotsc, i-1\} \cap \mathcal{T}(\boldsymbol{s}))} \lambda_{k,i}P_k}{N_{i}}  \Bigg) \right],
\end{equation}
which is  the reception rate of node $i$.
\end{thm}

\begin{proof}[Proof for Theorem~\ref{thm:df-gaussian-hd-krc}]
The rate is obtained from Theorem~\ref{thm:df-hd-k-rc} using independent Gaussian inputs (see Remark~\ref{remark:gaussian-optimal}).
\end{proof}

\subsection{The Optimal DF Scheduling Problem}\label{sec:hd-optimal-scheduling}
We define a schedule and an optimal DF schedule of the HD-MRC as follows. The schedule for a half-duplex network is defined as the probability mass function of all possible transmit/listen states, i.e., $p(\boldsymbol{s})$. An optimal DF schedule is the schedule that gives the highest DF rate.

Now, we formulate the optimal scheduling problem for the $D$-node phase-fading HD-MRC. An optimal schedule is a probability mass function $p^*(\boldsymbol{s})$ on $\{L,T\}^{D-2}$ that maximizes $R_\text{DF}$, i.e.,

\begin{multline}
p^*(\boldsymbol{s}) \in \argmax_{p(\boldsymbol{s})} \min_{i \in \{2,3,\dotsc,D\}} \\
\sum_{\substack{\boldsymbol{s} \in\{L,T\}^{D-2}\\\text{s.t. } i \in \mathcal{L}(\boldsymbol{s})\cup \{D\}}}\bigg[  p(\boldsymbol{s}) \Gamma \Bigg( \frac{\sum_{k \in \{1\} \cup ( \{2, \dotsc, i-1\} \cap \mathcal{T}(\boldsymbol{s}))} \lambda_{k,i}P_k}{N_{i}}  \Bigg)\Bigg].
 \label{eq:hd-optimal-schedule}
\end{multline}

\section{Solving Optimal Schedules for the HD-MRC}

\subsection{An Algorithm for General Maximin Problems}\label{sec:hd-minimax}
From the previous section, we know that the optimal DF scheduling problem for the phase-fading HD-MRC is a maximin optimization problem.
In this section, we propose an algorithm to solve a general class of maximin optimization problems. 
We will show how this algorithm can be simplified when used to solve the DF scheduling problem for the phase-fading HD-MRC. 

Consider the following maximin optimization problem:
\begin{equation}\label{eq:hd-general-maxmin}
\max_{\boldsymbol{p} \in \mathcal{G}} \min_{i \in \{1,2,\dotsc,K\}} \{ R_i(\boldsymbol{p}) \},
\end{equation}
where $\boldsymbol{p} = (p_0, p_1, \dotsc, p_{M-1})$ is a point in the closed convex set $\mathcal{G}$ in $\mathcal{R}^M$, and each $R_i(\boldsymbol{p})$ is a real and continuous function of $\boldsymbol{p}$. Both $M$ and $K$ are finite.
Our aim is to find the optimal $\boldsymbol{p}^*$ and $R^*$ given as follows:
\begin{subequations}
\begin{align}
\boldsymbol{p}^* &\in \argmax_{\boldsymbol{p} \in \mathcal{G}} \min_{i \in \{1,2,\dotsc,K\}}  R_i(\boldsymbol{p})\\
R^* &= \max_{\boldsymbol{p} \in \mathcal{G}} \min_{i \in \{1,2,\dotsc,K\}} R_i(\boldsymbol{p}) \triangleq \min_{i \in \{1,2,\dotsc,K\}} R_i(\boldsymbol{p}^*),
\end{align}
\end{subequations}
i.e., we want to find an optimal $\boldsymbol{p}$ that attains \eqref{eq:hd-general-maxmin}, denoted by $R^*$.

Solving the aforementioned maximin optimization problem corresponds to finding the optimal schedule and the best DF rate considered in this paper. In this application, $\mathcal{G}$ becomes the set of all feasible schedules, $\boldsymbol{p}$ is a schedule from $\mathcal{G}$, $R_i(\boldsymbol{p})$ is the reception rate for node $(i+1)$, $\boldsymbol{p}^*$ is an optimal schedule, and $R^*$ is the corresponding (maximum) DF rate.

For $\boldsymbol{p}_1, \boldsymbol{p}_2 \in \mathcal{G}$, we denote by $|\boldsymbol{p}_1 - \boldsymbol{p}_2| = \sqrt{ (p_{1,0}-p_{2,0})^2 + (p_{1,1} - p_{2,1})^2 + \dotsm + (p_{1,M-1} - p_{2,M-1})^2}$ the Euclidean distance between $\boldsymbol{p}_1$ and $\boldsymbol{p}_2$. 
\begin{defn}
For a set $\mathcal{B}=\{b_1,b_2,\dotsc,b_B\}$, we define the set of arguments of local maxima (denoted by $\arglmax$) as
\begin{equation}
\boldsymbol{p}' \in \arglmax_{\boldsymbol{p} \in \mathcal{G}} \{R_{b_1}(\boldsymbol{p}) = R_{b_2}(\boldsymbol{p}) = \dotsm = R_{b_B}(\boldsymbol{p}) \},
\end{equation}
iff
\begin{equation}
R_i(\boldsymbol{p}') = R_j(\boldsymbol{p}'), \quad \forall i,j \in \mathcal{B},
\end{equation}
and if there exists some positive $\epsilon$ such that for any $\boldsymbol{p} \in \mathcal{G}$ where $|\boldsymbol{p} - \boldsymbol{p}'| < \epsilon$, there exists at least one index $i \in \mathcal{B}$ such that
\begin{equation}
R_i(\boldsymbol{p}) \leq R_i(\boldsymbol{p'}).
\end{equation}
\end{defn}

In words, the aforementioned function says that if $\boldsymbol{p}$ is a local maximum for the set $\mathcal{B}$, we cannot simultaneously increase all $\{R_i(\cdot): i \in \mathcal{B}\}$ by ``moving'' $\boldsymbol{p}$ in any direction in $\mathcal{G}$.

Now, we establish a necessary condition for the optimal solution to \eqref{eq:hd-general-maxmin}.
\begin{lem}\label{lemma:algo1-1}
Let the solution to \eqref{eq:hd-general-maxmin} be $R^*$. There exists a set $\mathcal{B} \triangleq  \{b_1,b_2,\dotsc,b_B\} \subseteq \mathcal \{1,2,\dotsc,K\}$ and an optimal $\boldsymbol{p}^* \in \argmax_{\boldsymbol{p} \in \mathcal{G}} \min_{i \in \{1,2,\dotsc,K\}}  R_i(\boldsymbol{p})$ such that
\begin{align}
\boldsymbol{p}^* &\in \arglmax_{\boldsymbol{p}} \{R_{b_1}(\boldsymbol{p}) = R_{b_2}(\boldsymbol{p}) = \dotsm = R_{b_B}(\boldsymbol{p}) \} \label{eq:lemma-for-algo1}\\
R^*& = R_i(\boldsymbol{p}^*) < R_j(\boldsymbol{p}^*),\quad \forall i \in \mathcal{B}, \forall j \notin \mathcal{B}. \label{eq:lemma-for-algo1b}
\end{align}
\end{lem}

\begin{proof}[Proof of Lemma~\ref{lemma:algo1-1}]
By definition, $R^*$ takes the value of the minimum of $\{R_i(\boldsymbol{p}^*): i \in \{1,2,\dotsc,K\}\}$ for some $\boldsymbol{p}^*$. We let the set of all indices $\{i\}$ where $R_i(\boldsymbol{p}^*)$ equal $R^*$ be $\mathcal{B}$. This means we can always write $R^* = R_{b_1}(\boldsymbol{p}^*) = R_{b_2}(\boldsymbol{p}^*) = \dotsm = R_{b_B}(\boldsymbol{p}^*)$ for some $\boldsymbol{p}^*$, and  $R^* < R_j(\boldsymbol{p}^*)$, $\forall j \notin \mathcal{B}$. So, \eqref{eq:lemma-for-algo1b} must be true. Recall that $R_i(\boldsymbol{p})$ is a continuous function of $\boldsymbol{p}$.  Now, if \eqref{eq:lemma-for-algo1} is false, we can always increase all $\{R_i(\cdot): i \in \mathcal{B}\}$ simultaneously by some small positive amount by using a new $\boldsymbol{p}' = \boldsymbol{p}^* + \boldsymbol{\epsilon}$, where $\boldsymbol{\epsilon}$ is some length-$M$ vector with $|\boldsymbol{\epsilon}|$ being small, such that $\min_{i \in \mathcal{B}} R_i(\boldsymbol{p}') = R' > R^*$ and $R' < R_j(\boldsymbol{p}')$, $\forall j \notin \mathcal{B}$. This means if \eqref{eq:lemma-for-algo1} is false, $R^*$ cannot be the solution (contradiction).
\end{proof}

For $\boldsymbol{p}^*$ and $\mathcal{B}$ in Lemma~\ref{lemma:algo1-1} that attain the solution to  \eqref{eq:hd-general-maxmin}, we call the indices $\{b_1,b_2,\dotsc,b_B\}$ the bottleneck indices of the optimization problem.

Now, we present an algorithm to find the optimal $\boldsymbol{p}^*$ and $R^*$. This is a greedy algorithm that searches through all local maxima that satisfy the necessary condition in Lemma~\ref{lemma:algo1-1}.

\noindent \hrulefill
\begin{algorithm}[for the general maximin problem]\label{algorithm1}
\texttt{
\begin{enumerate}
\item Initialize $B=0$.
\item \label{hd-step-next-B} Set $B  \leftarrow B + 1$. For each subset/selection $\mathcal{B} = \{b_1,b_2,\dotsc,b_B\} \subseteq \{1,2,\dotsc,K\}$, find all $\boldsymbol{p}' \in \mathcal{G}$ such that
\begin{align}
&\boldsymbol{p}' \in \arglmax_{\boldsymbol{p}\in\mathcal{G}} \{R_{b_1}(\boldsymbol{p})= R_{b_2}(\boldsymbol{p})= \dotsm = R_{b_B}(\boldsymbol{p})\} \label{eq:algo-cond-1}\\
&R(\mathcal{B},\boldsymbol{p}') \triangleq R_{b_1}(\boldsymbol{p}') < R_j(\boldsymbol{p}'),\quad  \forall j \notin \mathcal{B}. \label{eq:algo-cond-2}
\end{align}
Repeat Step~\ref{hd-step-next-B} until $|\mathcal{B}|=K$.
\item The optimal solution is given by the largest $R(\mathcal{B},\boldsymbol{p}')$ and its corres-\\ ponding $\boldsymbol{p}'$.
\end{enumerate}
}
\end{algorithm}
\noindent\hrulefill\\

\subsection{A Simplified Algorithm for the Phase-Fading HD-MRC}

Now, we show that the algorithm proposed in the previous section can be simplified when used to solve the optimal scheduling problem for the phase-fading HD-MRC. Recall that the optimal scheduling for the $D$-node channel is to find $p^*(\boldsymbol{s})$ that attains
\begin{multline}\label{eq:opt-scheduling}
\max_{p(\boldsymbol{s})} \min_{i \in \{2,3, \dotsc, D\}} \\
\sum_{\substack{\boldsymbol{s} \in\{L,T\}^{D-2}\\\text{s.t. } i \in \mathcal{L}(\boldsymbol{s})\cup \{D\}}}\bigg[  p(\boldsymbol{s}) L \Bigg( \frac{\sum_{k \in \{1\} \cup ( \{2, \dotsc, i-1\} \cap \mathcal{T}(\boldsymbol{s}))} \lambda_{k,i}P_k}{N_{i}}  \Bigg)\Bigg].
\end{multline}
In the previous section, we have an algorithm to solve $\boldsymbol{p}$ that attains
\begin{equation}
\max_{\boldsymbol{p}\in\mathcal{G}} \min_{i \in \{1,2,\dotsc, K\}} R_i(\boldsymbol{p}),
\end{equation}
which is exactly the optimization problem in \eqref{eq:opt-scheduling} with the following substitution:
\begin{align}
&\bullet\; \boldsymbol{p} = (p_0, p_1 \dotsc, p_{2^{D-2}-1}),
\intertext{where each $p_j$ is uniquely mapped from each $p(\boldsymbol{s})$.}
&\bullet \; \mathcal{G} = \left\{\boldsymbol{p}: \sum_j p_j = 1, p_j \geq 0, \forall j \right\}.\\
&\bullet \; K = D-1.\\
&\bullet \; R_i(\boldsymbol{p}) = \sum_{\substack{\text{all states } \boldsymbol{s}\\ \text{s.t. node } (i+1) \text{ listens in state } \boldsymbol{s} }} \nonumber \\ &\quad\quad\quad\quad \left[  p(\boldsymbol{s}) \Gamma \left( \frac{\sum\limits_{\substack{\text{all nodes }k \leq i\\ \text{s.t. } i \text{ transmits in state } \boldsymbol{s}}} \lambda_{k,i+1}P_{k}}{N_{i+1}}  \right)\right].\label{eq:R-rate}
\end{align}
meaning that $R_i(\boldsymbol{p})$ is the reception rate of node $(i+1)$, i.e., $r_{i+1}$.

We further define the following:
\begin{defn}
For a set $\mathcal{B}=\{b_1,b_2,\dotsc,b_B\}$, we define the following global maximum function:
\begin{equation}
\boldsymbol{p}' \in \argmax_{\boldsymbol{p}\in\mathcal{G}} \{R_{b_1}(\boldsymbol{p}) = R_{b_2}(\boldsymbol{p}) = \dotsm = R_{b_B}(\boldsymbol{p}) \},
\end{equation}
iff
\begin{equation}
\boldsymbol{p}' \in \argmax_{\boldsymbol{p} \in \mathcal{G}} R_{b_1}(\boldsymbol{p}),
\end{equation}
subject to the constraint that
\begin{equation}
R_i(\boldsymbol{p}') = R_j(\boldsymbol{p}'), \quad \forall i,j \in \mathcal{B}.
\end{equation}
\end{defn}

Now, we will show that when applied to the phase-fading HD-MRC, Algorithm~\ref{algorithm1} can be simplified to the following:

\noindent \hrulefill
\begin{algorithm}[for the phase-fading HD-MRC]\label{algorithm2}
\texttt{
\begin{enumerate}
\item Initialize $B=0$.
\item  \label{hd-step-find-alpha-terminate-2} Set $B \leftarrow B+1$.
\begin{itemize}
\item If there exists some $\boldsymbol{p}' \in \mathcal{G}$ such that
\begin{align}
& \boldsymbol{p}' \in \argmax_{\boldsymbol{p}\in\mathcal{G}} \{R_{1}(\boldsymbol{p})= R_{2}(\boldsymbol{p})= \dotsm = R_{B}(\boldsymbol{p})\}, \label{eq:algo2-condition1}\\
& R_{1}(\boldsymbol{p}') < R_j(\boldsymbol{p}'), \quad \forall j \notin \mathcal{B}, \label{eq:algo2-condition2}
\end{align}
go to Step~\ref{hd-step-find-alpha-terminate-3}.
\item Else, repeat Step~\ref{hd-step-find-alpha-terminate-2}.
\end{itemize}
\item \label{hd-step-find-alpha-terminate-3} The algorithm terminates with an optimal schedule $\boldsymbol{p}^*=\boldsymbol{p}'$ and $R_\textrm{DF} = R_1(\boldsymbol{p}')$. 
\end{enumerate}
}
\end{algorithm}
\noindent\hrulefill\\

Comparing Algorithms~\ref{algorithm1} and \ref{algorithm2}, we need to prove the following for the latter:
\begin{enumerate}
\item  In Step~\ref{hd-step-next-B} of the algorithm, instead of considering every possible size-$B$ subset of $\{1,2,\dotsc,K\}$, we only need to consider $\{1,2,\dotsc,B\}$.
\item Instead of considering local maxima in \eqref{eq:algo-cond-1}, we only need to consider the global maximum in \eqref{eq:algo2-condition1}.
\item  We can stop the algorithm when a solution for \eqref{eq:algo2-condition1} and \eqref{eq:algo2-condition2} is found, i.e., we do not need to proceed to a larger $B$ if the terminating conditions are satisfied.
\end{enumerate}
The proof can be found in Appendix~\ref{appendix:algo2}. This also proves \cite[Conjecture 1]{ongjohnsonmotani09}.

\begin{remark}
Comparing Step~\ref{hd-step-next-B} in Algorithms~\ref{algorithm1} and \ref{algorithm2}, we need to consider $(2^{K}-1)$ combinations of $\mathcal{B}$ in the first algorithm for the general maximin optimization, but only need to consider at most $K$ combinations in the second algorithm for the phase-fading HD-MRC, where $K=D-1$.
\end{remark}

\begin{remark}
The set $\mathcal{B}$ that gives the solution corresponds to the set of nodes whose reception rates constrain the overall rate $R_\text{DF}$.
We see that the bottleneck nodes for the phase-fading HD-MRC are always the first $B$ relays, or all relays and the destination for the case of $B=D-1$. This means it is always the ``front'' nodes that constrain the overall DF rates.
\end{remark}

In Appendix~\ref{appendix-example}, we give an example showing how we can find an optimal schedule for the four-node phase-fading HD-MRC using Algorithm~\ref{algorithm2}, and obtain closed-form solutions under certain channel topologies.

\subsection{A Further-Simplified Algorithm for the rSNR-Degraded Phase-Fading HD-MRC}

We further simplify Algorithm~\ref{algorithm2} for the received signal-to-noise ratio (rSNR) degraded phase-fading HD-MRC, defined as follows.

\begin{defn}
A $D$-node MRC is said to be rSNR-degraded if and only if for all $i,j,k \in \{1,2,\dotsc,D\}$ where $i < j < k$, we have
\begin{equation}
\frac{\lambda_{i,j}P_i}{N_j} > \frac{\lambda_{i,k}P_i}{N_k}.\label{eq:rsnr-degraded}
\end{equation}
\end{defn}

In an rSNR-degraded channel, as the node index $k$ increases, the SNR from each lower index node to node $k$ becomes smaller.

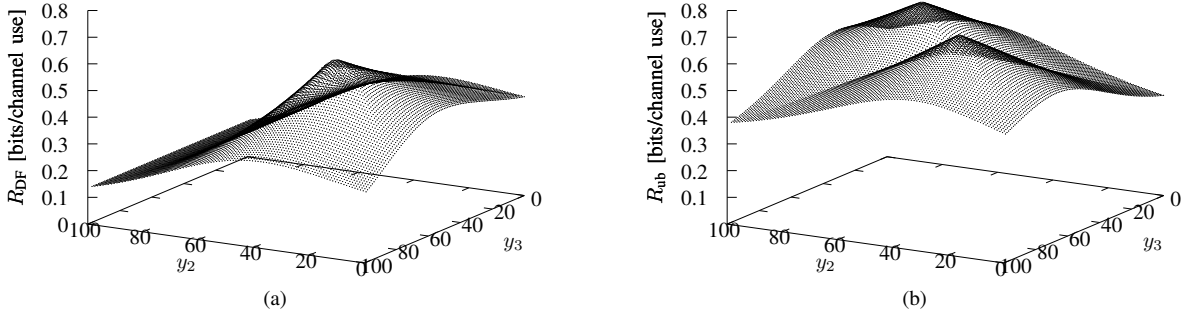
\begin{figure*}[t]
\centering
\subfloat[]
{\resizebox{0.4\linewidth}{!}{  
\begin{picture}(0,0)%
\includegraphics{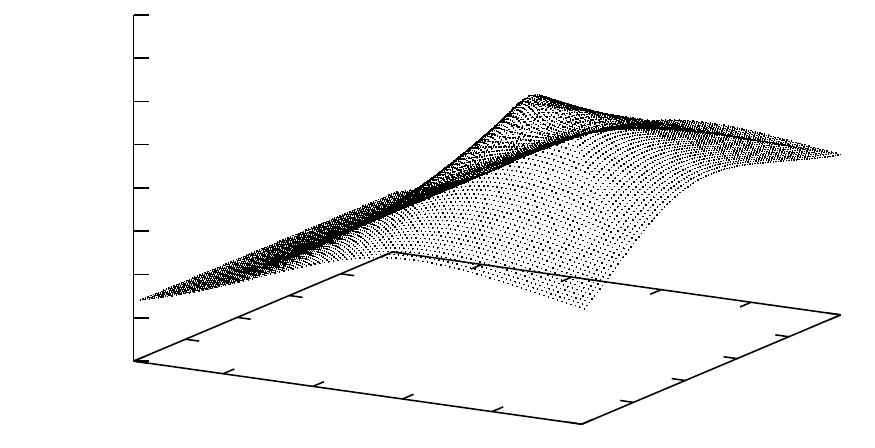}%
\end{picture}%
\setlength{\unitlength}{3947sp}%
\begingroup\makeatletter\ifx\SetFigFont\undefined%
\gdef\SetFigFont#1#2#3#4#5{%
  \reset@font\fontsize{#1}{#2pt}%
  \fontfamily{#3}\fontseries{#4}\fontshape{#5}%
  \selectfont}%
\fi\endgroup%
\begin{picture}(4185,2137)(1262,-3458)
\put(4003,-3443){\makebox(0,0)[b]{\smash{{\SetFigFont{10}{12.0}{\familydefault}{\mddefault}{\updefault} 0}}}}
\put(3573,-3382){\makebox(0,0)[b]{\smash{{\SetFigFont{10}{12.0}{\familydefault}{\mddefault}{\updefault} 20}}}}
\put(3143,-3322){\makebox(0,0)[b]{\smash{{\SetFigFont{10}{12.0}{\familydefault}{\mddefault}{\updefault} 40}}}}
\put(2712,-3261){\makebox(0,0)[b]{\smash{{\SetFigFont{10}{12.0}{\familydefault}{\mddefault}{\updefault} 60}}}}
\put(2282,-3200){\makebox(0,0)[b]{\smash{{\SetFigFont{10}{12.0}{\familydefault}{\mddefault}{\updefault} 80}}}}
\put(1852,-3140){\makebox(0,0)[b]{\smash{{\SetFigFont{10}{12.0}{\familydefault}{\mddefault}{\updefault} 100}}}}
\put(5363,-2897){\makebox(0,0)[b]{\smash{{\SetFigFont{10}{12.0}{\familydefault}{\mddefault}{\updefault} 0}}}}
\put(5114,-3001){\makebox(0,0)[b]{\smash{{\SetFigFont{10}{12.0}{\familydefault}{\mddefault}{\updefault} 20}}}}
\put(4866,-3106){\makebox(0,0)[b]{\smash{{\SetFigFont{10}{12.0}{\familydefault}{\mddefault}{\updefault} 40}}}}
\put(4617,-3211){\makebox(0,0)[b]{\smash{{\SetFigFont{10}{12.0}{\familydefault}{\mddefault}{\updefault} 60}}}}
\put(4369,-3316){\makebox(0,0)[b]{\smash{{\SetFigFont{10}{12.0}{\familydefault}{\mddefault}{\updefault} 80}}}}
\put(4120,-3421){\makebox(0,0)[b]{\smash{{\SetFigFont{10}{12.0}{\familydefault}{\mddefault}{\updefault} 100}}}}
\put(1754,-3103){\makebox(0,0)[rb]{\smash{{\SetFigFont{10}{12.0}{\familydefault}{\mddefault}{\updefault} 0}}}}
\put(1754,-2895){\makebox(0,0)[rb]{\smash{{\SetFigFont{10}{12.0}{\familydefault}{\mddefault}{\updefault} 0.1}}}}
\put(1754,-2687){\makebox(0,0)[rb]{\smash{{\SetFigFont{10}{12.0}{\familydefault}{\mddefault}{\updefault} 0.2}}}}
\put(1754,-2479){\makebox(0,0)[rb]{\smash{{\SetFigFont{10}{12.0}{\familydefault}{\mddefault}{\updefault} 0.3}}}}
\put(1754,-2271){\makebox(0,0)[rb]{\smash{{\SetFigFont{10}{12.0}{\familydefault}{\mddefault}{\updefault} 0.4}}}}
\put(1754,-2064){\makebox(0,0)[rb]{\smash{{\SetFigFont{10}{12.0}{\familydefault}{\mddefault}{\updefault} 0.5}}}}
\put(1754,-1856){\makebox(0,0)[rb]{\smash{{\SetFigFont{10}{12.0}{\familydefault}{\mddefault}{\updefault} 0.6}}}}
\put(1754,-1649){\makebox(0,0)[rb]{\smash{{\SetFigFont{10}{12.0}{\familydefault}{\mddefault}{\updefault} 0.7}}}}
\put(1754,-1441){\makebox(0,0)[rb]{\smash{{\SetFigFont{10}{12.0}{\familydefault}{\mddefault}{\updefault} 0.8}}}}
\put(1397,-2134){\rotatebox{90.0}{\makebox(0,0)[b]{\smash{{\SetFigFont{10}{12.0}{\familydefault}{\mddefault}{\updefault}$R_\text{DF}$ [bits/channel use]}}}}}
\put(2670,-3385){\makebox(0,0)[b]{\smash{{\SetFigFont{10}{12.0}{\familydefault}{\mddefault}{\updefault}$y_2$}}}}
\put(5215,-3219){\makebox(0,0)[b]{\smash{{\SetFigFont{10}{12.0}{\familydefault}{\mddefault}{\updefault}$y_3$}}}}
\put(1397,-2134){\rotatebox{90.0}{\makebox(0,0)[b]{\smash{{\SetFigFont{10}{12.0}{\familydefault}{\mddefault}{\updefault}$R_\text{DF}$ [bits/channel use]}}}}}
\end{picture}%
}
\label{fig:4mrc-ub}}
\hspace{1cm}
\subfloat[]
{\resizebox{0.4\linewidth}{!}{   
\begin{picture}(0,0)%
\includegraphics{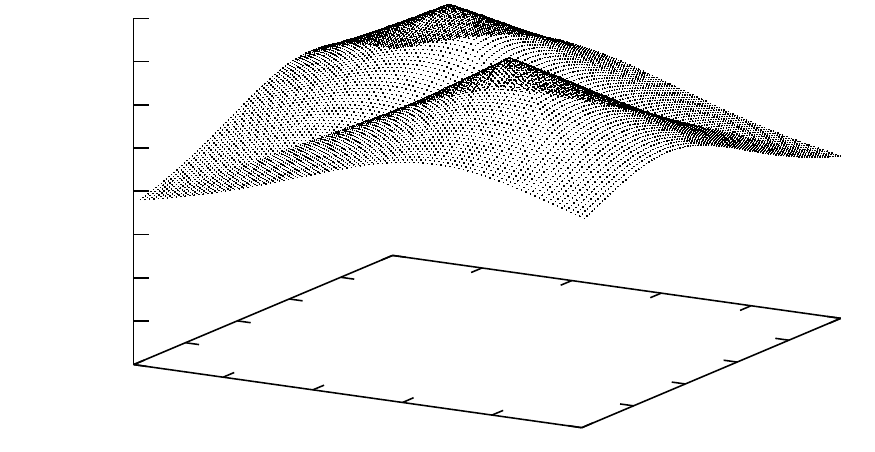}%
\end{picture}%
\setlength{\unitlength}{3947sp}%
\begingroup\makeatletter\ifx\SetFigFont\undefined%
\gdef\SetFigFont#1#2#3#4#5{%
  \reset@font\fontsize{#1}{#2pt}%
  \fontfamily{#3}\fontseries{#4}\fontshape{#5}%
  \selectfont}%
\fi\endgroup%
\begin{picture}(4185,2158)(1262,-3458)
\put(4003,-3443){\makebox(0,0)[b]{\smash{{\SetFigFont{10}{12.0}{\familydefault}{\mddefault}{\updefault} 0}}}}
\put(3573,-3382){\makebox(0,0)[b]{\smash{{\SetFigFont{10}{12.0}{\familydefault}{\mddefault}{\updefault} 20}}}}
\put(3143,-3322){\makebox(0,0)[b]{\smash{{\SetFigFont{10}{12.0}{\familydefault}{\mddefault}{\updefault} 40}}}}
\put(2712,-3261){\makebox(0,0)[b]{\smash{{\SetFigFont{10}{12.0}{\familydefault}{\mddefault}{\updefault} 60}}}}
\put(2282,-3200){\makebox(0,0)[b]{\smash{{\SetFigFont{10}{12.0}{\familydefault}{\mddefault}{\updefault} 80}}}}
\put(1852,-3140){\makebox(0,0)[b]{\smash{{\SetFigFont{10}{12.0}{\familydefault}{\mddefault}{\updefault} 100}}}}
\put(5363,-2897){\makebox(0,0)[b]{\smash{{\SetFigFont{10}{12.0}{\familydefault}{\mddefault}{\updefault} 0}}}}
\put(5114,-3001){\makebox(0,0)[b]{\smash{{\SetFigFont{10}{12.0}{\familydefault}{\mddefault}{\updefault} 20}}}}
\put(4866,-3106){\makebox(0,0)[b]{\smash{{\SetFigFont{10}{12.0}{\familydefault}{\mddefault}{\updefault} 40}}}}
\put(4617,-3211){\makebox(0,0)[b]{\smash{{\SetFigFont{10}{12.0}{\familydefault}{\mddefault}{\updefault} 60}}}}
\put(4369,-3316){\makebox(0,0)[b]{\smash{{\SetFigFont{10}{12.0}{\familydefault}{\mddefault}{\updefault} 80}}}}
\put(4120,-3421){\makebox(0,0)[b]{\smash{{\SetFigFont{10}{12.0}{\familydefault}{\mddefault}{\updefault} 100}}}}
\put(1754,-2895){\makebox(0,0)[rb]{\smash{{\SetFigFont{10}{12.0}{\familydefault}{\mddefault}{\updefault} 0.1}}}}
\put(1754,-2687){\makebox(0,0)[rb]{\smash{{\SetFigFont{10}{12.0}{\familydefault}{\mddefault}{\updefault} 0.2}}}}
\put(1754,-2479){\makebox(0,0)[rb]{\smash{{\SetFigFont{10}{12.0}{\familydefault}{\mddefault}{\updefault} 0.3}}}}
\put(1754,-2271){\makebox(0,0)[rb]{\smash{{\SetFigFont{10}{12.0}{\familydefault}{\mddefault}{\updefault} 0.4}}}}
\put(1754,-2064){\makebox(0,0)[rb]{\smash{{\SetFigFont{10}{12.0}{\familydefault}{\mddefault}{\updefault} 0.5}}}}
\put(1754,-1856){\makebox(0,0)[rb]{\smash{{\SetFigFont{10}{12.0}{\familydefault}{\mddefault}{\updefault} 0.6}}}}
\put(1754,-1649){\makebox(0,0)[rb]{\smash{{\SetFigFont{10}{12.0}{\familydefault}{\mddefault}{\updefault} 0.7}}}}
\put(1754,-1441){\makebox(0,0)[rb]{\smash{{\SetFigFont{10}{12.0}{\familydefault}{\mddefault}{\updefault} 0.8}}}}
\put(1397,-2134){\rotatebox{90.0}{\makebox(0,0)[b]{\smash{{\SetFigFont{10}{12.0}{\familydefault}{\mddefault}{\updefault}$R_\text{ub}$ [bits/channel use]}}}}}
\put(2670,-3385){\makebox(0,0)[b]{\smash{{\SetFigFont{10}{12.0}{\familydefault}{\mddefault}{\updefault}$y_2$}}}}
\put(5215,-3219){\makebox(0,0)[b]{\smash{{\SetFigFont{10}{12.0}{\familydefault}{\mddefault}{\updefault}$y_3$}}}}
\put(1397,-2134){\rotatebox{90.0}{\makebox(0,0)[b]{\smash{{\SetFigFont{10}{12.0}{\familydefault}{\mddefault}{\updefault}$R_\text{ub}$ [bits/channel use]}}}}}
\end{picture}%
}
\label{fig:4mrc-ub-diff}}
\caption{DF rates versus capacity upper bound for the four-node MRC with varying relay positions. (a) DF rates. (b) Capacity upper bound.}
\end{figure*}

\begin{remark}
The rSNR-degraded MRC is different from the degraded MRC defined by Xie and Kumar~\cite[Definition 2.3]{xiekumar03}. In the degraded MRC, $(X_1, \dotsc, X_{k-1}) \rightarrow (Y_k,X_k, \dotsc, X_{D-1}) \rightarrow (Y_{k+1}, \dotsc, Y_D)$ forms a Markov chain for each $k$. Applying this to the AWGN channel, roughly speaking, node $(k+1)$ receives the signals from nodes $\{1,2,\dotsc, k-1\}$ through $Y_k$. 
We see that the degraded channel is defined with respect to the channel from all upstream nodes (collectively) to a node, while the rSNR-degraded channel is defined for node pairs. 
Take $D=3$ for example. An example of the degraded condition is $Y_2 = \sqrt{\lambda_{1,2}}e^{j\theta(1,2)}X_1 + Z_2$ and $Y_3 = \sqrt{\lambda'_{1,3}}e^{j\theta'(1,3)}\left[\sqrt{\lambda_{1,2}}e^{j\theta(1,2)} X_1 + Z_2\right] + \sqrt{\lambda_{2,3}}e^{j\theta(2,3)}X_2 +  Z_3$, where $\lambda'_{1,3}$ and $\theta'(1,3)$ are known to node 3. The rSNR-degraded condition simply requires that $\frac{\lambda_{1,2}P_2}{N_2} > \frac{\lambda_{1,3}P_1}{N_3}$.
\end{remark}

For the rSNR-degraded phase-fading HD-MRC, the optimal scheduling problem can be found using the following algorithm.

\noindent \hrulefill
\begin{algorithm}[for the rSNR-degraded phase-fading HD-MRC]\label{algorithm3}
{\ }\newline
\indent\texttt{Find $\boldsymbol{p}^*$ that attains
\begin{equation}
\max_{\boldsymbol{p}\in\mathcal{G}}\{R_1(\boldsymbol{p}) = R_2(\boldsymbol{p}) = \dotsm = R_{K}(\boldsymbol{p})\}.\label{eq:simplified-algo}
\end{equation}
}
\end{algorithm}
\noindent\hrulefill\\

\begin{proof}[Proof for Algorithm~\ref{algorithm3}]
Suppose that an optimal schedule and $R_\text{DF}$ are given by 
\begin{align}
\boldsymbol{p}^* &\in \max_{\boldsymbol{p}} \{R_1(\boldsymbol{p}) = R_2(\boldsymbol{p}) = \dotsm = R_B (\boldsymbol{p}) \}\label{eq:condition-3-M}\\
R_\text{DF} = R_B(\boldsymbol{p}^*) &< R_j(\boldsymbol{p}^*), \forall j \notin \mathcal{B},\label{eq:condition-four-M}
\end{align}
for some $\mathcal{B}=\{1,2,\dotsc,B\} \subset  \{1,2,\dotsc, K\}$, i.e., $B < K$.

From the proof for Algorithm~\ref{algorithm2}, we have established that node $(B+1)$ must always listen in the optimal schedule.
So, the reception rate of node $(B+1)$ is
\begin{align}
r_{B+1} &= R_B(\boldsymbol{p}^*) \nonumber\\ &= \sum_{\text{all states } \boldsymbol{s}} \left[ p(\boldsymbol{s}) \Gamma \left( \frac{\sum\limits_{\substack{\text{all nodes }i \leq B\\ \text{s.t. }i \text{ transmits in state } \boldsymbol{s}}}\lambda_{i,B+1}P_i}{N_{B+1}}   \right) \right].
\end{align}

Now, the reception rate of node $(B+2)$ is
\begin{subequations}
\begin{align}
r_{B+2} & = R_{B+1}(\boldsymbol{p}^*)\\ &\leq \sum_{\text{all states } \boldsymbol{s}} \left[ p(\boldsymbol{s}) \Gamma \left( \frac{\sum\limits_{\substack{\text{all nodes }i \leq B\\ \text{s.t. }i \text{ transmits in state } \boldsymbol{s}}}\lambda_{i,B+2}P_i}{N_{B+2}}   \right) \right]\label{eq:max-R_M+2}\\
& < R_B(\boldsymbol{p}^*).\label{eq:R_M+2-less-thanR_M+1}
\end{align}
\end{subequations}
Eqn.~\eqref{eq:max-R_M+2} is derived because node $(B+1)$ always listens; so nodes $(B+2)$ and $(B+1)$ both decode from the same set of nodes in each state. An equality is achieved when node $(B+2)$ also always listens in all states. Eqn.~\eqref{eq:R_M+2-less-thanR_M+1} follows from the definition of rSNR-degraded MRC.

So, condition~\eqref{eq:condition-four-M} cannot be satisfied for any $\mathcal{B} \subset \{1,2,\dotsc, K\}$, and therefore $B < K$ will not give an optimal schedule. This means the only possible case is $\mathcal{B}=\{1,2,\dotsc,K\}$.
\end{proof}

\begin{remark}
Compared to the original maximin optimization in \eqref{eq:hd-optimal-schedule}, we only need to solve the max optimization in \eqref{eq:simplified-algo} to obtain an optimal schedule for the  rSNR-degraded phase-fading HD-MRC.
\end{remark}

\begin{remark}
For the rSNR-degraded phase-fading HD-MRC, the bottleneck nodes are always all the relays and the destination, regardless of the network topology. This means in an optimal schedule, the relays coordinate their transmit/listen states so that no single relay or a subset of relays will constrain the overall DF rate.
\end{remark}

\emph{A note on complexity:} In Algorithm~\ref{algorithm3}, we only need to consider one case for $\mathcal{B}$, compared to
$(2^{D-1}-1)$ cases in Algorithm~\ref{algorithm1}, However, equating the equalities in \eqref{eq:simplified-algo} has a time complexity of $O(2^{D-2})$, where $2^{D-2}$ is the length of the vector $\boldsymbol{p}$. A heuristic algorithm~\cite{wangongmotani09} can be used to reduced the dimension of $\boldsymbol{p}$, but doing that does not guarantee optimality.

\section{Numerical Computations}

\subsection{The Four-Node Phase-Fading HD-MRC}\label{sec:computations}

In this section, we demonstrate that Algorithm~\ref{algorithm2} provides insights into how one would operate the network. We consider the four-node phase-fading HD-MRC where the nodes are placed on a straight line. The node coordinates are as follows: node 1 $(0,0)$, node 2 $(0,y_2)$, node 3 $(0,y_3)$, and node 4 $(0,100)$. We set the channel parameters to be $\kappa=1$, $\eta=2$, and fix the per-symbol power constraints to be $P_i=10$ for all $i \in \{1,2,3\}$, and the noise variance at the receivers to be $N_j=0.01$ for all $j \in \{2,3,4\}$. Define $p_0=p(L,L)$, $p_1=p(L,T)$, $p_2=p(T,L)$, and $p_3=p(T,T)$.

Fig.~\ref{fig:4mrc-ub} shows DF rates for the four-node phase-fading HD-MRC with varying relay positions using Algorithm~\ref{algorithm2}, and Fig.~\ref{fig:4mrc-ub-diff} shows the capacity upper bound. Note that at any $y_2$ and $y_3$, the DF rate is strictly below the capacity upper bound.

\begin{figure*}[t]
\centering
\subfloat[]
{\resizebox{0.37\linewidth}{!}{  
\begin{picture}(0,0)%
\includegraphics{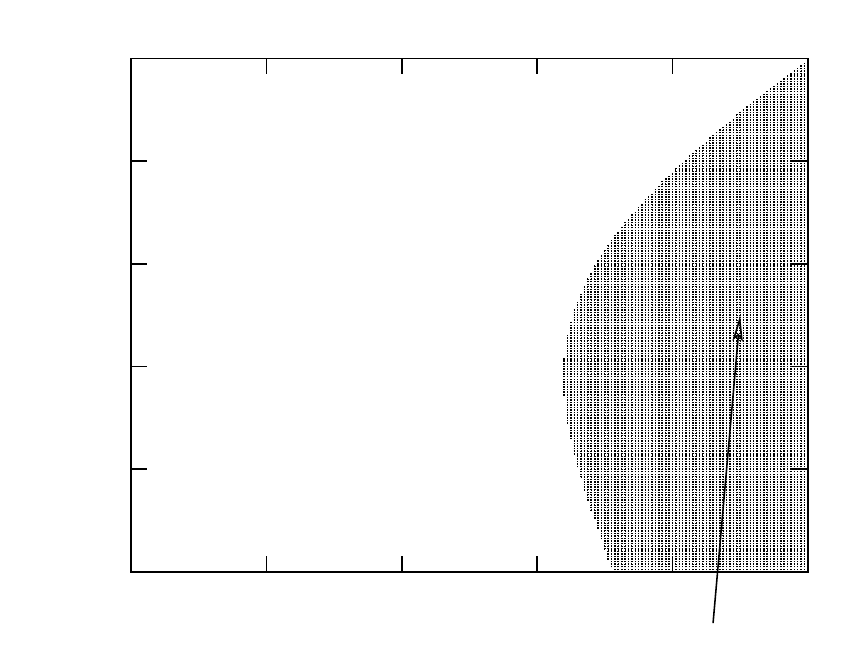}%
\end{picture}%
\setlength{\unitlength}{3947sp}%
\begingroup\makeatletter\ifx\SetFigFont\undefined%
\gdef\SetFigFont#1#2#3#4#5{%
  \reset@font\fontsize{#1}{#2pt}%
  \fontfamily{#3}\fontseries{#4}\fontshape{#5}%
  \selectfont}%
\fi\endgroup%
\begin{picture}(4023,3176)(1334,-3999)
\put(1888,-3623){\makebox(0,0)[rb]{\smash{{\SetFigFont{10}{12.0}{\familydefault}{\mddefault}{\updefault} 0}}}}
\put(1888,-3130){\makebox(0,0)[rb]{\smash{{\SetFigFont{10}{12.0}{\familydefault}{\mddefault}{\updefault} 20}}}}
\put(1888,-2637){\makebox(0,0)[rb]{\smash{{\SetFigFont{10}{12.0}{\familydefault}{\mddefault}{\updefault} 40}}}}
\put(1888,-2144){\makebox(0,0)[rb]{\smash{{\SetFigFont{10}{12.0}{\familydefault}{\mddefault}{\updefault} 60}}}}
\put(1888,-1651){\makebox(0,0)[rb]{\smash{{\SetFigFont{10}{12.0}{\familydefault}{\mddefault}{\updefault} 80}}}}
\put(1888,-1158){\makebox(0,0)[rb]{\smash{{\SetFigFont{10}{12.0}{\familydefault}{\mddefault}{\updefault} 100}}}}
\put(1963,-3748){\makebox(0,0)[b]{\smash{{\SetFigFont{10}{12.0}{\familydefault}{\mddefault}{\updefault} 0}}}}
\put(2613,-3748){\makebox(0,0)[b]{\smash{{\SetFigFont{10}{12.0}{\familydefault}{\mddefault}{\updefault} 20}}}}
\put(3263,-3748){\makebox(0,0)[b]{\smash{{\SetFigFont{10}{12.0}{\familydefault}{\mddefault}{\updefault} 40}}}}
\put(3912,-3748){\makebox(0,0)[b]{\smash{{\SetFigFont{10}{12.0}{\familydefault}{\mddefault}{\updefault} 60}}}}
\put(4562,-3748){\makebox(0,0)[b]{\smash{{\SetFigFont{10}{12.0}{\familydefault}{\mddefault}{\updefault} 80}}}}
\put(5212,-3748){\makebox(0,0)[b]{\smash{{\SetFigFont{10}{12.0}{\familydefault}{\mddefault}{\updefault} 100}}}}
\put(1481,-2329){\rotatebox{90.0}{\makebox(0,0)[b]{\smash{{\SetFigFont{10}{12.0}{\familydefault}{\mddefault}{\updefault}$y_3$}}}}}
\put(3587,-3935){\makebox(0,0)[b]{\smash{{\SetFigFont{10}{12.0}{\familydefault}{\mddefault}{\updefault}$y_2$}}}}
\put(3587,-970){\makebox(0,0)[b]{\smash{{\SetFigFont{10}{12.0}{\familydefault}{\mddefault}{\updefault}MRC: $(0,0)$, $(0,y_2)$, $(0,y_3)$, $(0,100)$}}}}
\put(4562,-3919){\makebox(0,0)[lb]{\smash{{\SetFigFont{10}{12.0}{\familydefault}{\mddefault}{\updefault}$|\mathcal{B}|=1$}}}}
\put(2938,-2144){\makebox(0,0)[lb]{\smash{{\SetFigFont{10}{12.0}{\familydefault}{\mddefault}{\updefault}$|\mathcal{B}|=3$}}}}
\end{picture}%
}
\label{fig:cases-vs-relay-positions}}
\hspace{2cm}
\subfloat[]
{\resizebox{0.37\linewidth}{!}{  
\begin{picture}(0,0)%
\includegraphics{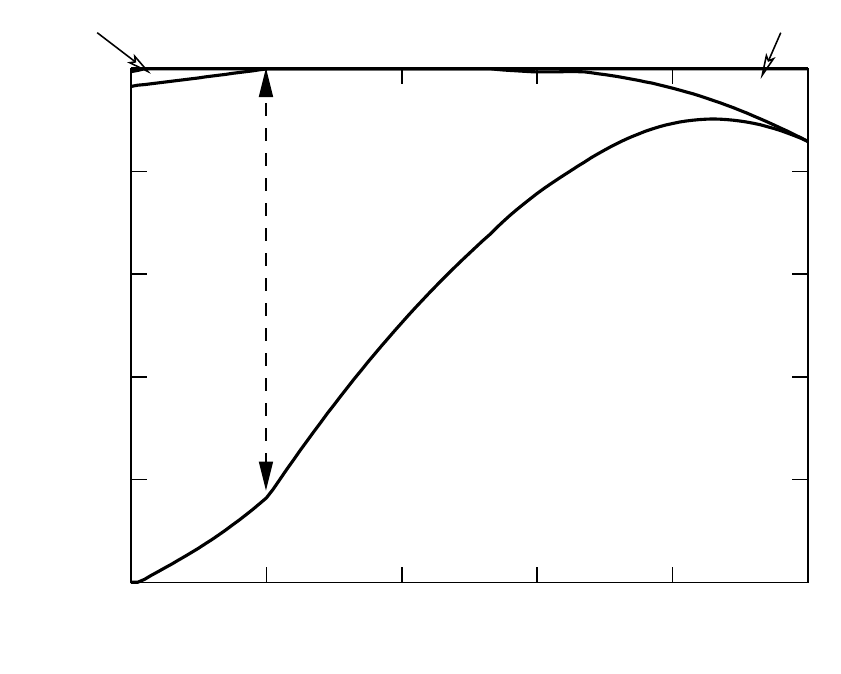}%
\end{picture}%
\setlength{\unitlength}{3947sp}%
\begingroup\makeatletter\ifx\SetFigFont\undefined%
\gdef\SetFigFont#1#2#3#4#5{%
  \reset@font\fontsize{#1}{#2pt}%
  \fontfamily{#3}\fontseries{#4}\fontshape{#5}%
  \selectfont}%
\fi\endgroup%
\begin{picture}(4023,3223)(1334,-3999)
\put(2551,-2760){\rotatebox{90.0}{\makebox(0,0)[lb]{\smash{{\SetFigFont{10}{12.0}{\familydefault}{\mddefault}{\updefault}$p^*(L,T)$ at $y_3=20$}}}}}
\put(1888,-3130){\makebox(0,0)[rb]{\smash{{\SetFigFont{10}{12.0}{\familydefault}{\mddefault}{\updefault} 0.2}}}}
\put(1888,-2637){\makebox(0,0)[rb]{\smash{{\SetFigFont{10}{12.0}{\familydefault}{\mddefault}{\updefault} 0.4}}}}
\put(1888,-2144){\makebox(0,0)[rb]{\smash{{\SetFigFont{10}{12.0}{\familydefault}{\mddefault}{\updefault} 0.6}}}}
\put(1888,-1651){\makebox(0,0)[rb]{\smash{{\SetFigFont{10}{12.0}{\familydefault}{\mddefault}{\updefault} 0.8}}}}
\put(1888,-1158){\makebox(0,0)[rb]{\smash{{\SetFigFont{10}{12.0}{\familydefault}{\mddefault}{\updefault} 1}}}}
\put(1963,-3748){\makebox(0,0)[b]{\smash{{\SetFigFont{10}{12.0}{\familydefault}{\mddefault}{\updefault} 0}}}}
\put(2613,-3748){\makebox(0,0)[b]{\smash{{\SetFigFont{10}{12.0}{\familydefault}{\mddefault}{\updefault} 20}}}}
\put(3263,-3748){\makebox(0,0)[b]{\smash{{\SetFigFont{10}{12.0}{\familydefault}{\mddefault}{\updefault} 40}}}}
\put(3912,-3748){\makebox(0,0)[b]{\smash{{\SetFigFont{10}{12.0}{\familydefault}{\mddefault}{\updefault} 60}}}}
\put(4562,-3748){\makebox(0,0)[b]{\smash{{\SetFigFont{10}{12.0}{\familydefault}{\mddefault}{\updefault} 80}}}}
\put(5212,-3748){\makebox(0,0)[b]{\smash{{\SetFigFont{10}{12.0}{\familydefault}{\mddefault}{\updefault} 100}}}}
\put(1481,-2329){\rotatebox{90.0}{\makebox(0,0)[b]{\smash{{\SetFigFont{10}{12.0}{\familydefault}{\mddefault}{\updefault}$p(s)$}}}}}
\put(3587,-3935){\makebox(0,0)[b]{\smash{{\SetFigFont{10}{12.0}{\familydefault}{\mddefault}{\updefault}$y_3$}}}}
\put(3512,-970){\makebox(0,0)[b]{\smash{{\SetFigFont{10}{12.0}{\familydefault}{\mddefault}{\updefault}MRC: $(0,0)$, $(0,66)$, $(0,y_3)$, $(0,100)$}}}}
\put(3588,-2884){\makebox(0,0)[lb]{\smash{{\SetFigFont{10}{12.0}{\familydefault}{\mddefault}{\updefault}$(L,L)$}}}}
\put(1476,-911){\makebox(0,0)[lb]{\smash{{\SetFigFont{10}{12.0}{\familydefault}{\mddefault}{\updefault}$(T,L)$}}}}
\put(4887,-911){\makebox(0,0)[lb]{\smash{{\SetFigFont{10}{12.0}{\familydefault}{\mddefault}{\updefault}$(T,L)$}}}}
\put(2928,-1744){\makebox(0,0)[lb]{\smash{{\SetFigFont{10}{12.0}{\familydefault}{\mddefault}{\updefault}$(L,T)$}}}}
\put(1888,-3623){\makebox(0,0)[rb]{\smash{{\SetFigFont{10}{12.0}{\familydefault}{\mddefault}{\updefault} 0}}}}
\end{picture}%
}
\label{fig:y1-66}}
\caption{Optimal schedules for the four-node phase-fading HD-MRC: Graph (a) shows $|\mathcal{B}|$ (the number of bottleneck nodes) for which an optimal schedule can be found. Graph (b) shows optimal schedules for $y_2=66$ and varying node 3's position: for each $y_3$, the optimal schedule consists of probabilities of the four states, where the probabilities of the states take the height of their respective regions at $y_3$ [the probability of the state $(T,T)$ is always zero].}
\end{figure*}
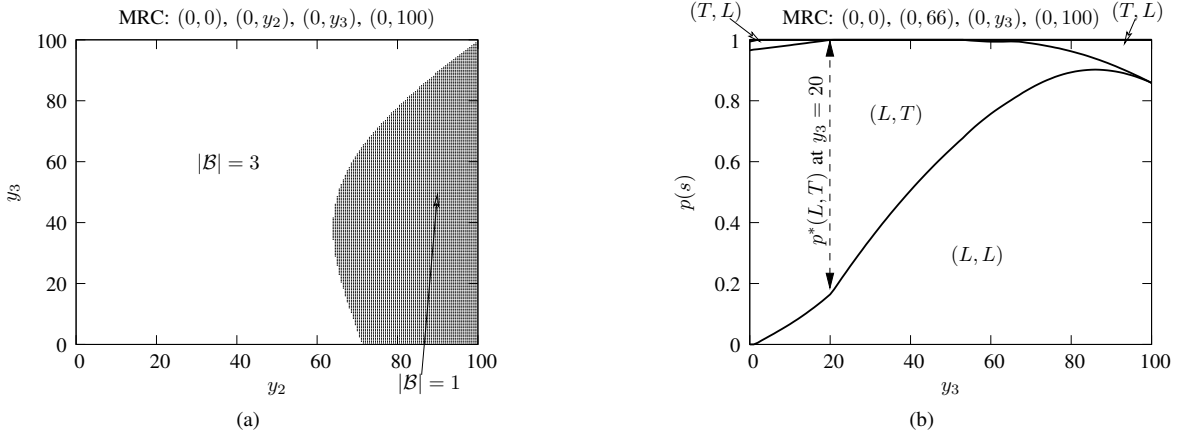

Fig.~\ref{fig:cases-vs-relay-positions} shows which case in Algorithm~\ref{algorithm2}---whether $|\mathcal{B}|$ is $1$, $2$, or $3$---corresponds to the optimal schedule, for the node positions $y_1=0$, $y_4=100$, and $0 \leq y_2, y_3 \leq 100$. From \eqref{eq:hd-case-3-4mrc-a} in the derivation of the optimal schedule in Appendix~\ref{appendix-example}, we know that  $|\mathcal{B}|=1$ has an optimal schedule only when node 2 is further from node 1 than node 3 is from node 1, or $y_2 \geq y_3$. From Fig.~\ref{fig:cases-vs-relay-positions}, we see that $|\mathcal{B}|=2$ has no optimal schedule for these network topologies. In addition, if the nodes in the network are arranged such $0 < y_2 < y_3 < 100$, then we only need to consider $|\mathcal{B}|=3$.

Fig.~\ref{fig:y1-66} shows optimal schedules for varying node 3's position  $(0,y_3)$, while fixing the rest of the nodes' positions at $y_1=0$, $y_2=66$, and $y_4=100$. From Fig.~\ref{fig:cases-vs-relay-positions}, we see that for $y_2=66$ and $20 \leq y_3 < 53$, $\boldsymbol{p}^*$ is found in the case $|\mathcal{B}|=1$, meaning that the optimal schedules are $(p_0^*, 1-p_0^*,0,0)$. We can also see this from Fig.~\ref{fig:y1-66} that we only need to operate the network in two states, $(L,T)$ and $(L,L)$, for $20 \leq y_3 < 53$. This means node 2 does not need to transmit, and only node 3 needs to toggle between the transmitting and the listening modes. Meanwhile, for $y_2=66$, $0 \leq y_3 < 20$ or $53 \leq y_3 < 100$, three states are used in the optimal schedule.

\subsection{rSNR-degraded FD-MRCs and HD-MRCs}\label{sec:fd-vs-hd}
Next, we compare DF rates for the $D$-node phase-fading FD-MRCs and HD-MRCs. We consider the line topology with two settings:
\begin{enumerate}
\item the coordinate of node $i$ is $(i,0)$;
\item the coordinate of node $i$ is $(\frac{20i}{D-1},0)$.
\end{enumerate}
Note that the distance between the source and the destination increases in the first case, and remains fixed in the second case.

We fix the per-symbol power constraints $P_i = 10$ for all $i$ and the noise variance at the receivers $N_i = 1$ for all $i$. As these MRCs are rSNR-degraded, we can calculate DF rates for the half-duplex channels using Algorithm~\ref{algorithm3}.

\begin{figure}
\begin{center}
\resizebox{0.9\linewidth}{!}{   
\begin{picture}(0,0)%
\includegraphics{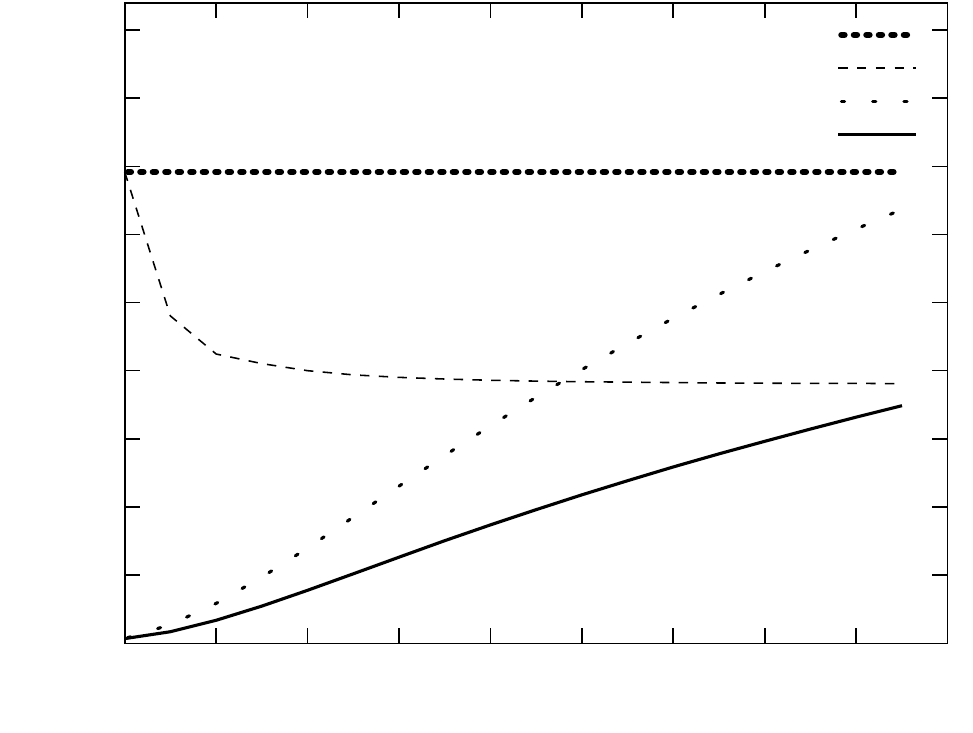}%
\end{picture}%
\setlength{\unitlength}{3947sp}%
\begingroup\makeatletter\ifx\SetFigFont\undefined%
\gdef\SetFigFont#1#2#3#4#5{%
  \reset@font\fontsize{#1}{#2pt}%
  \fontfamily{#3}\fontseries{#4}\fontshape{#5}%
  \selectfont}%
\fi\endgroup%
\begin{picture}(4655,3516)(1240,-4015)
\put(1763,-3648){\makebox(0,0)[rb]{\smash{{\SetFigFont{10}{12.0}{\familydefault}{\mddefault}{\updefault} 0}}}}
\put(1763,-3321){\makebox(0,0)[rb]{\smash{{\SetFigFont{10}{12.0}{\familydefault}{\mddefault}{\updefault} 0.5}}}}
\put(1763,-2994){\makebox(0,0)[rb]{\smash{{\SetFigFont{10}{12.0}{\familydefault}{\mddefault}{\updefault} 1}}}}
\put(1763,-2667){\makebox(0,0)[rb]{\smash{{\SetFigFont{10}{12.0}{\familydefault}{\mddefault}{\updefault} 1.5}}}}
\put(1763,-2339){\makebox(0,0)[rb]{\smash{{\SetFigFont{10}{12.0}{\familydefault}{\mddefault}{\updefault} 2}}}}
\put(1763,-2012){\makebox(0,0)[rb]{\smash{{\SetFigFont{10}{12.0}{\familydefault}{\mddefault}{\updefault} 2.5}}}}
\put(1763,-1685){\makebox(0,0)[rb]{\smash{{\SetFigFont{10}{12.0}{\familydefault}{\mddefault}{\updefault} 3}}}}
\put(1763,-1358){\makebox(0,0)[rb]{\smash{{\SetFigFont{10}{12.0}{\familydefault}{\mddefault}{\updefault} 3.5}}}}
\put(1763,-1031){\makebox(0,0)[rb]{\smash{{\SetFigFont{10}{12.0}{\familydefault}{\mddefault}{\updefault} 4}}}}
\put(1763,-704){\makebox(0,0)[rb]{\smash{{\SetFigFont{10}{12.0}{\familydefault}{\mddefault}{\updefault} 4.5}}}}
\put(1838,-3773){\makebox(0,0)[b]{\smash{{\SetFigFont{10}{12.0}{\familydefault}{\mddefault}{\updefault} 2}}}}
\put(2277,-3773){\makebox(0,0)[b]{\smash{{\SetFigFont{10}{12.0}{\familydefault}{\mddefault}{\updefault} 4}}}}
\put(2716,-3773){\makebox(0,0)[b]{\smash{{\SetFigFont{10}{12.0}{\familydefault}{\mddefault}{\updefault} 6}}}}
\put(3155,-3773){\makebox(0,0)[b]{\smash{{\SetFigFont{10}{12.0}{\familydefault}{\mddefault}{\updefault} 8}}}}
\put(3594,-3773){\makebox(0,0)[b]{\smash{{\SetFigFont{10}{12.0}{\familydefault}{\mddefault}{\updefault} 10}}}}
\put(4032,-3773){\makebox(0,0)[b]{\smash{{\SetFigFont{10}{12.0}{\familydefault}{\mddefault}{\updefault} 12}}}}
\put(4471,-3773){\makebox(0,0)[b]{\smash{{\SetFigFont{10}{12.0}{\familydefault}{\mddefault}{\updefault} 14}}}}
\put(4910,-3773){\makebox(0,0)[b]{\smash{{\SetFigFont{10}{12.0}{\familydefault}{\mddefault}{\updefault} 16}}}}
\put(5349,-3773){\makebox(0,0)[b]{\smash{{\SetFigFont{10}{12.0}{\familydefault}{\mddefault}{\updefault} 18}}}}
\put(5788,-3773){\makebox(0,0)[b]{\smash{{\SetFigFont{10}{12.0}{\familydefault}{\mddefault}{\updefault} 20}}}}
\put(1387,-2049){\rotatebox{90.0}{\makebox(0,0)[b]{\smash{{\SetFigFont{10}{12.0}{\familydefault}{\mddefault}{\updefault}$R_\text{DF}$ [bits/channel use]}}}}}
\put(3813,-3960){\makebox(0,0)[b]{\smash{{\SetFigFont{10}{12.0}{\familydefault}{\mddefault}{\updefault}Number of nodes, $D$}}}}
\put(5188,-727){\makebox(0,0)[rb]{\smash{{\SetFigFont{10}{12.0}{\familydefault}{\mddefault}{\updefault}Full Duplex, $v_i=i$}}}}
\put(5188,-886){\makebox(0,0)[rb]{\smash{{\SetFigFont{10}{12.0}{\familydefault}{\mddefault}{\updefault}Half Duplex, $v_i=i$}}}}
\put(5188,-1045){\makebox(0,0)[rb]{\smash{{\SetFigFont{10}{12.0}{\familydefault}{\mddefault}{\updefault}Full Duplex, $v_i=20i/(D-1)$}}}}
\put(5188,-1204){\makebox(0,0)[rb]{\smash{{\SetFigFont{10}{12.0}{\familydefault}{\mddefault}{\updefault}Half Duplex, $v_i=20i/(D-1)$}}}}
\end{picture}%
}
\caption{$R_\text{DF}$ vs $D$ for phase-fading MRCs, where the coordinate of node $i$ is $(v_i,0)$}
\label{fig:fd-vs-hd}
\end{center}
\end{figure}
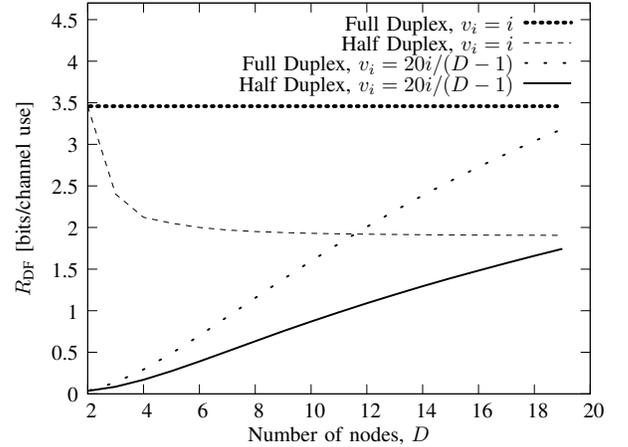

Fig.~\ref{fig:fd-vs-hd} shows DF rates for these two networks as the number of relays increases. As expected, DF rates for the half-duplex network are upper-bounded by those for the full-duplex network~\cite{kramergastpar04}. The DF rates are significantly reduced in the half-duplex networks compared to the full-duplex counterparts as the number of relays increases. In the first setting where all adjacent nodes are separated a unit distance apart, the full-duplex DF rate remains constant regardless of the network size (or equivalently, the source-destination distance) as the link from node 1 to node 2 is always the bottleneck. This is because any added node sees a higher SNR compared to its upstream nodes, and hence the reception rate of the added node is always higher. This is not true when the nodes cannot listen and transmit simultaneously. As the relays do transmit all the time (otherwise their reception rate will be zero), the added node may not have a higher reception rate compared to its upstream nodes. We illustrate this by comparing two networks of $D=2$ and $D=3$. For $D=2$, the source (node 1) always transmits and the destination (node 2) always listen. So, the DF rates for the HD-MRC, denoted by $R_{\text{DF}[D=2]}^\text{h}$, equals that for the FD-MRC, denoted by $R_{\text{DF}[D=2]}^\text{f}$. Next, consider the FD-MRC with $D=3$. The reception rate at node 2 (which can listen all the time) remains the same at $R_{\text{DF}[D=2]}^\text{f}$. As nodes 1 and 2 always transmit, node 3 sees a strictly higher SNR, and has a strictly higher reception rate. So, the DF rate for $D=3$, which is constrained by the reception rates of nodes 2 and 3, is the same as that for $D=2$, i.e.,  $R_{\text{DF}[D=3]}^\text{f} = R_{\text{DF}[D=2]}^\text{f}$. Now consider the HD-MRC with $D=3$. The relay (node 2) allocates its time to transmit and to listen, and hence its reception rate is lower than $R_{\text{DF}[D=2]}^\text{h}$. In addition, node 3 receives transmission from node 1 (which is further away) all the time, but from node 2 only a fraction of the time. From Algorithm~\ref{algorithm3}, we know that node 2 picks the fraction of transmission time such that its reception rate equals that at node 3. Hence the DF rate for the HD-MRC is lower than that of the FD-MRC for $D=3$, i.e.,  $R_{\text{DF}[D=3]}^\text{h} < R_{\text{DF}[D=3]}^\text{f}$
This example illustrates a fundamental difference between the full-duplex and the half-duplex networks.

\section{Reflections}

We have investigated achievable rates and optimal schedules for the decode-forward (DF) coding strategy for the half-duplex multiple-relay channel (HD-MRC). The code construction in this paper differs from that of the traditional half-duplex network. Traditionally, a node transmits an entire codeword in one state duration. The transmission time for a codeword is shorter than the duration that the nodes stay in a particular state. In any state, data can only be sent to and from fixed sets of nodes. Thus, the order of the states becomes important as data need to be routed from the source to the destination through the relays. Hence, the order of the states and the time fraction of the states are parameters of the optimization.

Inspired by the coding scheme for the half-duplex single-relay channel, where the codewords of the source and of the relay are split into two parts (one part when the relay listen, and another part when the relay transmits), we approached the scheduling problem in this perspective. Each codeword cycles through all states, meaning that the states change during one codeword duration and they are independent of how the data flow. The advantage of this approach is that the order of the states is not important as far as the rate is concerned and we have shown that only the time fraction (or the probability) of different states needs to be optimized.

We have proposed an algorithm to solve a general class of maximin optimization problems. When applied to the phase-fading HD-MRC, the algorithm can be simplified. Further simplification to max optimizations has been shown to be possible for the rSNR-degraded phase-fading HD-MRC. 
The simplified algorithms have allowed us to obtain closed-form solutions for optimal schedules for certain network topologies.

From the algorithms, we realized that the bottleneck nodes of the phase-fading HD-MRC are always the first $B$ relays for some $1 \leq B \leq D-2$, or all the relays plus the destination. Furthermore, the bottleneck nodes for the rSNR-degraded phase-fading HD-MRC are always all the relays plus the destination, independent of the channel topology. This shows that in the half-duplex network, the transmit/listen state attempts to balance the transmission bottleneck of the half-duplex network.

We have also discovered that, unlike the full-duplex case, where DF achieves the cut-set upper bound and hence the capacity when the relays are close to the source~\cite{kramergastpar04}, DF in the half duplex case does not achieve the cut-set bound.

\appendices

\section{Proof of Theorem~\ref{thm:hd-mrc-upper-bound}} \label{appendix:ub}

Consider the capacity upper bound of the half-duplex multiterminal network~\cite[Corollary 2]{khojastepoursabharwal03} with the following additional constraints due to the HD-MRC model considered in the paper:
\begin{itemize}
\item We only consider the rate from node 1 to node $D$, which is $R$.
\item There are at most $2^{D-2}$ states, and the time fraction of the network in state $\boldsymbol{s}$ is $p(\boldsymbol{s})$.
\end{itemize}
We have the following upper bound on the maximum achievable rate $R$ (i.e., the capacity $C$):
\begin{equation}
C \leq \sup_{p(\boldsymbol{s})} \min_\mathcal{Q} \sum_{\boldsymbol{s} \in \{L,T\}^{D-2}} p(\boldsymbol{s})  I(X_1,X_\mathcal{Q}; Y_{\mathcal{Q}^c},Y_D | X_{\mathcal{Q}^c}, \boldsymbol{S} ), \label{eq:fd-cut-set}
\end{equation}
for some $p(x_1, x_2, \dotsc, x_{D-1}|\boldsymbol{s})$, where the supremum is taken over all possible schedule $p(\boldsymbol{s})$, and the minimization is taken over all possible $\mathcal{Q} \subseteq \mathcal{R}$. Here, $\mathcal{Q}^c = \mathcal{R} \setminus \mathcal{Q}$, and $\mathcal{R}$ is the set of all relays.

In state $\boldsymbol{s}$, we have the following:
\begin{itemize}
\item For $X_{\mathcal{Q}}$: $X_{\mathcal{Q} \cap \mathcal{T}(\boldsymbol{s})}$ are transmitted, and $X_{\mathcal{Q} \cap \mathcal{L}(\boldsymbol{s})} = \tilde{x}_{\mathcal{Q} \cap \mathcal{L}(\boldsymbol{s})}$ are known to all nodes,
\item For $Y_{\mathcal{Q}^c}$: $Y_{\mathcal{Q}^c\cap\mathcal{L}(\boldsymbol{s})}$ are received, and $Y_{\mathcal{Q}^c\cap\mathcal{T}(\boldsymbol{s})}=\tilde{y}_{\mathcal{Q}^c\cap\mathcal{T}(\boldsymbol{s})}$ are known to all nodes.
\item For $X_{\mathcal{Q}^c}$: $X_{\mathcal{Q}^c \cap \mathcal{T}(\boldsymbol{s})}$ are transmitted, and $X_{\mathcal{Q}^c \cap \mathcal{L}(\boldsymbol{s})} = \tilde{x}_{\mathcal{Q} ^c\cap \mathcal{L}(\boldsymbol{s})}$ are known to all nodes.
\end{itemize}
Note that for four random variables $A$, $B$, $C$, and $D$, if $B$ takes on a fixed value $e$, we have
\begin{align}
I(A,B=e;C|D) &= I(B=e;C|D) + I(A;C|D,B=e)\nonumber \\ & = I(A;C|D,B=e)\\
I(A;B=e,C|D) &= I(A;B=e|D) + I(A;C|D,B=e) \nonumber \\ & = I(A;C|D,B=e).
\end{align}
So the mutual information term on the right-hand side of \eqref{eq:fd-cut-set} becomes $I(X_1,X_{\mathcal{Q}\cap\mathcal{T}(\boldsymbol{s})}; Y_{\mathcal{Q}^c\cap\mathcal{L}(\boldsymbol{s})},Y_D |$ $X_{\mathcal{Q}^c \cap \mathcal{T}(\boldsymbol{s})},X_{\mathcal{Q}^c \cap \mathcal{L}(\boldsymbol{s})} = \tilde{x}_{\mathcal{Q}^c \cap \mathcal{L}(\boldsymbol{s})}, X_{\mathcal{Q} \cap \mathcal{L}(\boldsymbol{s})} = \tilde{x}_{\mathcal{Q} \cap \mathcal{L}(\boldsymbol{s})},$ $Y_{\mathcal{Q}^c\cap\mathcal{T}(\boldsymbol{s})}=\tilde{y}_{\mathcal{Q}^c\cap\mathcal{T}(\boldsymbol{s})} )$. Rearranging the variables, we get $ I\Big(X_1,X_{\mathcal{Q}\cap\mathcal{T}(\boldsymbol{s})}; Y_{\mathcal{Q}^c\cap\mathcal{L}(\boldsymbol{s})},Y_D | X_{\mathcal{Q}^c\cap\mathcal{T}(\boldsymbol{s})}, X_{\mathcal{L}(\boldsymbol{s})} = \tilde{x}_{\mathcal{L}(\boldsymbol{s})},$ $Y_{\mathcal{Q}^c\cap\mathcal{T}(\boldsymbol{s})}=\tilde{y}_{\mathcal{Q}^c\cap\mathcal{T}(\boldsymbol{s})} \Big)$.

Under the half-duplex constraints, we also have
\begin{multline}
p(x_1, x_2, \dotsc, x_{D-1}|\boldsymbol{s}) \\ = \prod_{k\in\mathcal{L}(\boldsymbol{s})} \delta(x_k=\tilde{x}_k) p(x_1,x_{\mathcal{T}(\boldsymbol{s})}|x_{\mathcal{L}(\boldsymbol{s})}, \boldsymbol{s}). \label{eq:hd-all-possible-input}
\end{multline}

Since we can choose  any schedule $p(\boldsymbol{s})$ and any input distribution $p(x_1,x_2,\dotsc,x_{D-1}|\boldsymbol{s})$ of the form in \eqref{eq:hd-all-possible-input}, we get Theorem~\ref{thm:hd-mrc-upper-bound}.
{\ } \hfill $\blacksquare$

\section{Proof of Theorem~\ref{thm:df-hd-k-rc}} \label{appendix:df}
As mentioned, the random coding proof for the FD-MRC~\cite[Section  IV.2]{xiekumar03} is adapted to our context of the HD-MRC. In addition, we use the idea of splitting the codewords into different parts, one for each state~\cite{hostmadsen02}. We consider $n$ transmissions or channel uses as a block and send $(M-D+2)$ messages over $M$ blocks. For any $n$, we can choose a sufficiently large $M$ such that the effective rate $\frac{nR(M-D+2)}{nM}$ can be made as close to $R$ as desired. As there are $(D-2)$ relays, the total number of states $\boldsymbol{s}$ is $2^{D-2}$. Let the states be $\boldsymbol{s}_0$, $\boldsymbol{s}_1$, $\dotsc$, $\boldsymbol{s}_{2^{D-2}-1}$, and the probabilities of the states be $p_0$, $p_1$, $\dotsc$, $p_{2^{D-2}-1}$ respectively. First, we generate codebooks for the HD-MRC.

\noindent {\bf Codebook generation:}

\begin{enumerate}
\item Fix, $p(\boldsymbol{s})$, i.e., $p_0$, $p_1$, $\dotsc$, $p_{2^{D-2}-1}$. We assume that $n_i \triangleq np_i$ are integers.
\item 
Fix the states for all channel uses in each block as follows: 
 for the first $n_0=np_0$ channel uses, the state is fixed at $\boldsymbol{s}_0$, for the next $n_1=np_1$ channel uses, the state is $\boldsymbol{s}_1$, and so on.
\item Fix the probability mass functions
\begin{multline}
p(x_1,x_2,\dotsc, x_{D-1}|\boldsymbol{s})=p(x_{D-1}|\boldsymbol{s})p(x_{D-2}|x_{D-1},\boldsymbol{s})\\ \dotsm p(x_1|x_2,x_3,\dotsc,x_{D-1},\boldsymbol{s}),\label{eq:channel-input-prob}
\end{multline}
one for each state $\boldsymbol{s}$. 
\item For the $p(x_{D-1}|\boldsymbol{s})$ chosen in \eqref{eq:channel-input-prob}, generate $2^{nR}$ independently and identically distributed (i.i.d.) $n$-sequences $\boldsymbol{x}_{D-1} \in \mathcal{X}_{D-1}^{n}$ according to
\begin{equation}
p(\boldsymbol{x}_{D-1}) = \prod_{i=0}^{2^{D-2}-1} \prod_{t=\sum_{j=-1}^{i-1}n_j+1}^{\sum_{j=-1}^{i}n_j} p(x_{D-1}[t]|\boldsymbol{s}_i),
\end{equation}
where $n_{-1}=0$, and $x_{i}[t]$ denotes the $t$-th input from node $i$ into the channel.
Index the codewords as $\boldsymbol{x}_{D-1}(w_{D-1})$, $w_{D-1} \in \{1,2,\dotsc,2^{nR}\}$.
\item\label{step-repeat} Repeat the following steps for $k= D-2, D-3, \dotsc,1$:
\begin{enumerate}
\item For the $p(x_k|x_{k+1},x_{k+2},\dotsc,x_{D-1},\boldsymbol{s})$ chosen in \eqref{eq:channel-input-prob}, and for each
\begin{align}
& \boldsymbol{x}_{k+1:D-1}(w_{k+1},w_{k+2},\dotsc,w_{D-1})\nonumber \\  & \triangleq
\Big(\boldsymbol{x}_{k+1}(w_{k+1}|w_{k+2},\dotsc,w_{D-1}), \nonumber \\ &\quad\quad \boldsymbol{x}_{k+2}(x_{k+2}|w_{k+3},\dotsc,w_{D-1}),\nonumber \\ &\quad\quad  \dotsc,  \boldsymbol{x}_{D-1}(w_{D-1}) \Big), \nonumber
\end{align}
generate $2^{nR}$ conditionally independent $n$-sequences $\boldsymbol{x}_k \in \mathcal{X}_k^n$ according to
\begin{align*}
& p(\boldsymbol{x}_k | \boldsymbol{x}_{k+1:D-1}(w_{k+1},w_{k+2},\dotsc,w_{D-1}))\\ &= \prod_{i=0}^{2^{D-2}-1} \prod_{t=\sum_{j=-1}^{i-1}n_j+1}^{\sum_{j=-1}^{i}n_j} \nonumber \\ & \quad p \Big( x_k[t] \Big| x_{k+1:D-1}(w_{k+1},w_{k+2},\dotsc,w_{D-1})[t],\boldsymbol{s}_i \Big),
\end{align*}
where
\begin{align}
& x_{k+1:D-1}(w_{k+1},w_{k+2},\dotsc,w_{D-1}) [t] \nonumber \\  & = \Big(x_{k+1}(w_{k+1}|w_{k+2},\dotsc,w_{D-1})[t],
\nonumber \\ &\quad\quad  x_{k+2}(x_{k+2}|w_{k+3},\dotsc,w_{D-1})[t],  \nonumber \\ &\quad\quad\dotsc, x_{D-1}(w_{D-1})[t]\Big).\nonumber
\end{align}
Index the codewords $\boldsymbol{x}_k(w_k|w_{k+1},w_{k+2},$ $\dotsc,w_{D-1})$, $w_k \in \{1,2,\dotsc,2^{nR}\}$.
\item Repeat Step~\ref{step-repeat}(a) with $k \leftarrow (k-1)$.
\end{enumerate}
\end{enumerate}

\noindent {\bf Encoding and Decoding:}

Let the message sequence be $w(1)$, $w(2)$, $\dotsc$, $w(M-D+2)$. These messages will be sent in $M$ blocks of transmissions. Let node $k$'s estimate of message $w(m)$ be $\hat{w}_k(m)$.
\begin{enumerate}
\item In block $m \in [1,M]$, node $k \in [1,D-1]$ sends $\boldsymbol{x}_k \big(\hat{w}_k(m-k+1) \big| \hat{w}_k(m-k), \dotsc, \hat{w}_k(m-D+2)\big)$.
We set $\hat{w}_1(l)=w(l)$ for all $l$, and $\hat{w}_k(m)=1$ for $m \leq 0$ and $m \geq M-D+3$. The transmission from node $k$ in block $m$ depends on the messages that it has decoded in the past $(D-k)$ blocks.
\item In block $m \in [1,M]$, node $k \in \{2,\dotsc,D\}$ receives $\boldsymbol{y}_k(m)$.
\item At the end of block $m$, having already decoded $\hat{w}_k(1)$, $\hat{w}_k(2)$, $\dotsc$, $\hat{w}_k(m-k+1)$ in the previous blocks,  node $k$ declares $\hat{w}_k(m-k+2) = w$,
if there exists a unique $w$ such that the following is true for all $j \in \{0, 1, \dotsc, k-2\}$.
\begin{align}\label{eq:typical-w}
\Big\{ & \boldsymbol{x}_{k-1-j}(w|\hat{w}_k(m-k+1),\dots, \hat{w}_k(m-j-D+2)), \nonumber \\
& \boldsymbol{x}_{k-j}(\hat{w}_k(m-k+1)| \hat{w}_k(m-k)), \dots, \nonumber \\
&\quad\quad\;\; \hat{w}_k(m-j-D+2)),\nonumber \\
& \dotsc, \boldsymbol{x}_{D-1}(\hat{w}_k(m-j-D+2),\boldsymbol{y}_k(m-j), \underline{\boldsymbol{s}} \Big\} \nonumber \\
& \in \mathcal{A}_\epsilon^{(n)}(X_{k-1-j},X_{k-j},\dotsc,X_{D-1},Y_k,\boldsymbol{S}),
\end{align}
where  $\underline{\boldsymbol{s}}\triangleq(\boldsymbol{s}[1],\boldsymbol{s}[2],\dotsc,\boldsymbol{s}[n])$ is the vector of states for the block of $n$ channel uses, and $\mathcal{A}_\epsilon^{(n)}(Z_1,Z_2,\dotsc,Z_M)$ is the set of $\epsilon$-typical $n$-sequences $(\boldsymbol{z}_1,\boldsymbol{z}_2,\dotsc,\boldsymbol{z}_M)$ defined as follows.
\end{enumerate}

\begin{defn}
An  $\epsilon$-typical $n$-sequences $\mathcal{A}_\epsilon^{(n)}(Z_1,Z_2,$ $\dotsc,Z_M)$ with respect to a distribution $p(z_1,z_2,\dotsc, z_M)$ on $\mathcal{Z}_1 \times \mathcal{Z}_2 \times \dotsm \times \mathcal{Z}_M$ is the set of sequences $(\boldsymbol{z}_1, \boldsymbol{z}_2, \dotsc, \boldsymbol{z}_M) \in \mathcal{Z}_1^n \times \mathcal{Z}_2^n \times \dotsm \times \mathcal{Z}_M^n$ such that
\begin{equation}
\left\vert -\frac{1}{n}\log p(\boldsymbol{u}) - H(U)  \right\vert < \epsilon,\quad \forall U \subseteq \{Z_1,Z_2,\dotsc,Z_M\},
\end{equation}
where $p(\boldsymbol{u}) = \prod_{i=1}^n p(u[t])$.
\end{defn}

\noindent {\bf Probability of Error Analysis:}
We define the following:
\begin{align}
A_\text{correct}(m) & \triangleq \text{the event that no decoding error occurs at any} \nonumber \\ & \quad\; \text{node in the first $m$ blocks},\\
\mathcal{W}_{k,j}(m) &\triangleq \left\{\text{all } w \in \{1,2,\dotsc,2^{nR}\} \text{ that satisfy \eqref{eq:typical-w}} \right\},\\
\mathcal{W}_k(m) &\triangleq \bigcap_{j=0}^{k-2}\mathcal{W}_{k,j}(m).\label{eq:typical-w-2}
\end{align}
$\mathcal{W}_k(m)$ is the set containing node $k$'s estimate(s) for the message $w(m-k+2)$. This decoding step is done in block $m$.

\begin{figure*}[!t] 
\normalsize
\setcounter{mytempeqncnt}{\value{equation}}
\setcounter{equation}{56}
\begin{subequations}
\begin{align}
&I(X_1, X_2, \dotsc, X_{k-1};Y_k | X_k, X_{k+1}\dotsc, X_{D-1},\boldsymbol{S})\nonumber\\
&= \sum_{\boldsymbol{s} \in \{\boldsymbol{s}_0,\boldsymbol{s}_1,\dotsc,\boldsymbol{s}_{2^{D-1}-1}\}} p(\boldsymbol{s}) I(X_1, X_2, \dotsc, X_{k-1};Y_k | X_k, X_{k+1}\dotsc, X_{D-1},\boldsymbol{S}=\boldsymbol{s}) \label{eq:a}\\
&=  \sum_{\substack{\boldsymbol{s} \in \{\boldsymbol{s}_0,\boldsymbol{s}_1,\dotsc,\boldsymbol{s}_{2^{D-1}-1}\}\\ \text{s.t. } k \in \mathcal{L}(\boldsymbol{s}) \cup \{D\}}} p(\boldsymbol{s}) I(X_1, X_2, \dotsc, X_{k-1};Y_k | X_k, X_{k+1}, \dotsc, X_{D-1},\boldsymbol{S}=\boldsymbol{s})\label{eq:known-reception}\\
&=  \sum_{\substack{\boldsymbol{s} \in \{\boldsymbol{s}_0,\boldsymbol{s}_1,\dotsc,\boldsymbol{s}_{2^{D-1}-1}\}\\ \text{s.t. } k \in \mathcal{L}(\boldsymbol{s}) \cup \{D\} }} p(\boldsymbol{s}) I(X_1, X_{\{2, \dotsc,k-1\} \cap \mathcal{T}(\boldsymbol{s})};Y_k | X_{\{k,\dotsc,D-1\} \cap \mathcal{T}(\boldsymbol{s})}, \{ X_j=\tilde{x}_j: j \in \mathcal{L}(\boldsymbol{s}) \}, \boldsymbol{S}=\boldsymbol{s}).\label{eq:known-transmission}
\end{align}
\end{subequations}
\setcounter{equation}{\value{mytempeqncnt}}
\hrulefill
\vspace*{4pt} \end{figure*}

Using a property of $\epsilon$-typical sequences~\cite[Theorem 15.2.1]{coverthomas06}, with sufficiently large $n$, we can bound the probability that the correct message $w(m-k+2)$ is not in the set $\mathcal{W}_{k,j}(m)$ to be arbitrarily small, i.e.,
\begin{equation}
\Pr \{ w(m-k+2) \notin \mathcal{W}_{k,j}(m) |A_\text{correct}(m-1) \} < \epsilon. \label{eq:by-aep}
\end{equation}
Following~\cite[Theorem 15.2.3]{coverthomas06} with the substitutions
\begin{equation}
S_1 = \boldsymbol{X}_{k-1-j}(w'|\dotsm), S_2 = \boldsymbol{Y}_k, S_3 = (\boldsymbol{X}_{k-j},\dotsc, \boldsymbol{X}_{D-1}),
\end{equation}
for some $w' \neq w(m-k+2)$, and with a sufficiently large $n$, we can again bound the probability that the wrong message $w'$ is in the set $\mathcal{W}_{k,j}(m)$ to be arbitrarily small, i.e.,
\begin{multline}
\Pr \{ w' \in \mathcal{W}_{k,j}(m) |A_\text{correct}(m-1) \} \\ < 2^{-n(I(X_{k-1-j};Y_k | X_{k-j},\dotsc, X_{D-1},S)-6\epsilon)},\label{eq:follows-from}
\end{multline}
for all  $j \in \{0, 1, \dotsc, k-2 \}$.

From \eqref{eq:by-aep}, we can bound the probability that the correct message $w(m-k+2)$ is not in node $k$'s set of message estimates to be arbitrarily small, i.e.,
\begin{subequations}
\begin{align}
& \Pr \{  w(m-k+2) \notin \mathcal{W}_k(m) |A_\text{correct}(m-1) \} \nonumber \\
&\leq \sum_{j=0}^{k-2} \Pr \{ w(m-k+2) \notin \mathcal{W}_{k,j}(m) |A_\text{correct}(m-1) \} \nonumber \\
& < (k-1)\epsilon \label{eq:small-e-2}\\
& \leq (D-1)\epsilon.\label{eq:small-e-3}
\end{align}
\end{subequations}
Eqn.~\eqref{eq:small-e-2} is true for any $\epsilon > 0$ if a sufficiently large $n$ is chosen.

From \eqref{eq:follows-from}, we can also bound the probability that any wrong message $w'$ is in node $k$'s set of message estimates to be arbitrarily small, i.e.,
\begin{subequations}
\begin{align}
&\Pr \Big\{ \text{Any wrong message } w' \neq w(m-k+2) \in \mathcal{W}_k(m) \Big| \nonumber \\ &\quad\quad A_\text{correct}(m-1) \Big\}\nonumber \\
&\leq \sum_{\substack{w' \in \{1,2,\dotsc,2^{nR}\}\\ \text{s.t. } w' \neq w(m-k+2)}} \Pr \Big\{ w' \in \mathcal{W}_k(m) \Big| A_\text{correct}(m-1) \Big\}\\
&= \sum_{\substack{w' \in \{1,2,\dotsc,2^{nR}\}\\ \text{s.t. } w' \neq w(m-k+2)}} \Pr \Big\{ w' \in \mathcal{W}_{k,j}(m), \forall j \in [0,k-2] \Big| \nonumber \\ &\quad\quad\quad\quad\quad\quad\quad\quad\quad\quad A_\text{correct}(m-1) \Big\}\\
&\leq \sum_{\substack{w' \in \{1,2,\dotsc,2^{nR}\}\\ \text{s.t. } w' \neq w(m-k+2)}} \prod_{j=0}^{k-2} \Pr \Big\{ w' \in \mathcal{W}_{k,j}(m) \Big| A_\text{correct}(m-1) \Big\}\\
&< (2^{nR}-1) 2^{-n\sum_{j=0}^{k-2}(I(X_{k-1-j};Y_k | X_{k-j},\dotsc, X_{D-1},S)-6\epsilon)}\label{eq:follow}\\
& = (2^{nR}-1) 2^{-n(I(X_1,  \dotsc, X_{k-1};Y_k | X_k, X_{k+1},\dotsc, X_{D-1},S)-6(k-1)\epsilon)}\\
& < 2^{-n[I(X_1, X_2, \dotsc, X_{k-1};Y_k | X_k, X_{k+1}\dotsc, X_{D-1},S)-6(k-1)\epsilon - R ]}\\
&< \epsilon_1,\label{eq:small-e}
\end{align}
\end{subequations}
where \eqref{eq:follow} follows from \eqref{eq:follows-from}.

\addtocounter{equation}{1}

Eqn.~\eqref{eq:small-e} is true for any $\epsilon_1 > 0$ if a sufficiently large $n$ is chosen and if the following condition is satisfied:
\begin{align}
R &< I(X_1, X_2, \dotsc, X_{k-1};Y_k | X_k, X_{k+1}\dotsc, X_{D-1},\boldsymbol{S}) \nonumber \\ & \quad - 6(k-1)\epsilon, \label{eq:rate-condition}
\end{align}
where $6(k-1)\epsilon \rightarrow 0$ as $n$ increases.

So, if $n$ is sufficiently large and if \eqref{eq:rate-condition} is satisfied, combining \eqref{eq:small-e-3} and \eqref{eq:small-e}, the probability of node $k$ making a decoding error in block $m$, conditioned on the event that there was no decoding error in the previous $(m-1)$ blocks, can be made as small as desired, i.e.,
\begin{multline}
\Pr \Big\{ \text{Node $k$ wrongly decodes $w(m-k+2)$ in block $m$} \big| \\ A_\text{correct}(m-1) \Big\} < (D-1)\epsilon + \epsilon_1.\label{eq:replace-by-D-ref}
\end{multline}

Now, consider all the nodes and all the blocks. Using \eqref{eq:replace-by-D-ref} and the same arguments as in \cite{xiekumar03}, the probability of any node making a decoding error in any block can be made as small as desired, for any chosen $M$ and $D$, if $n$ is sufficiently large and if \eqref{eq:rate-condition} is satisfied.

The mutual information term in \eqref{eq:rate-condition} can be written as \eqref{eq:a} shown at the top of the page.
For nodes that transmits, i.e.,  $j \notin \mathcal{L}(\boldsymbol{s}) \cup \{D\}$, we have $Y_j = \tilde{y}_j$ and $I(X_1, X_2, \dotsc, X_{j-1};Y_j=\tilde{y}_j | X_j, X_{j+1}, \dotsc, X_{D-1},\boldsymbol{S}=\boldsymbol{s})=0$. This gives \eqref{eq:known-reception}. For nodes that listen, i.e., $j \in \mathcal{L}(\boldsymbol{s})$, we have $X_j=\tilde{x}_j$. This gives \eqref{eq:known-transmission}.

Since the result holds for any state and input distributions in the form $p(x_1,x_2,\dotsc,x_{D-1},\boldsymbol{s})=p(\boldsymbol{s})\prod_{k\in\mathcal{L}(\boldsymbol{s})} \delta(x_k=\tilde{x}_k) p(x_1,x_{\mathcal{T}(\boldsymbol{s})}|x_{\mathcal{L}(\boldsymbol{s})},\boldsymbol{s})$, we get Theorem~\ref{thm:df-hd-k-rc}. \hfill $\blacksquare$

\section{Proof for Algorithm~\ref{algorithm2}}\label{appendix:algo2}

In this section, we will prove that when applied to the phase-fading HD-MRC, Algorithm~\ref{algorithm1} can be simplified to Algorithm~\ref{algorithm2}. The proof consists of three parts, which are the differences between the two algorithms. Without loss of generality we assume that $b_1 < b_2 < \dotsm < b_B$.

\subsection{To prove that we only need to consider one set, i.e., $\{1,2,\dotsc,B\}$, for each $B$}

First, we consider Step~\ref{hd-step-next-B} in Algorithm~\ref{algorithm1} for $|\mathcal{B}|=1$. All possible cases for $\mathcal{B}$ are $\{1\}$, $\{2\}$, $\dotsc$, $\{K\}$. Suppose that an optimal solution exists for $\mathcal{B} = \{j\}$ where $j \neq 1$ with an optimal $\boldsymbol{p}^*$ where
\begin{align}
&\boldsymbol{p}^* \in \arglmax_{\boldsymbol{p}} R_j(\boldsymbol{p}), \\
&R_j(\boldsymbol{p}^*) < R_k(\boldsymbol{p}^*),\quad  \forall k \neq j. \label{eq:algo-check-condition}
\end{align}
Recall that $R_j(\cdot)$ is the reception rate of node $(j+1)$.
The optimal $\boldsymbol{p}^*$ must have zero probability for all the states in which node $(j+1)$ transmits. Otherwise, we can decrease the probability of a state $\boldsymbol{s}'$ in which node $(j+1)$ transmits by some small amount $\epsilon$ and increase the probability of the state $\boldsymbol{s}''$ by $\epsilon$ (where $\boldsymbol{s}'$ and $\boldsymbol{s}''$ only differ in node $(j+1)$'s transmit/listen mode). Let the new schedule be $\boldsymbol{p}''$. The new schedule will have a higher reception rate for node $(j+1)$, $R_j(\cdot)$, by some small amount while maintaining $R_j(\boldsymbol{p}'') < R_k(\boldsymbol{p}'')$, for all $k \neq j$, since all \{$R_i(\boldsymbol{p})\}$ are continuous functions of $\boldsymbol{p}$. So, in the optimal schedule, node $(j+1)$ always listens.

In addition, suppose that $\boldsymbol{p}^*$ has a non-zero probability, $p(\boldsymbol{s}') = p_a>0$, for state $\boldsymbol{s}'=(L,\dotsc,L,\dotsc)$ in which both nodes 2 and $(j+1)$ listen, and define $p(\boldsymbol{s}'') = p_b$, where $\boldsymbol{s}''$ equals $\boldsymbol{s}'$ except for the first position, i.e., in state $\boldsymbol{s}''=(T,\dotsc,L,\dotsc)$, node 2 transmits and node $(j+1)$ listens. We now show that $\boldsymbol{p}^*$ cannot be an optimal schedule. Let $\boldsymbol{p}''$ be the same schedule as $\boldsymbol{p}^*$ except that it has $p(\boldsymbol{s}')=p_a - \epsilon$ and $p(\boldsymbol{s}'') = p_b + \epsilon$, for some small positive $\epsilon$. Since $j > 1$, we have $R_j(\boldsymbol{p}'') > R_j(\boldsymbol{p}^*)$. Also, since $\epsilon$ is a small positive value, the inequality in \eqref{eq:algo-check-condition} is still maintained, i.e., $R_j(\boldsymbol{p}'') < R_k(\boldsymbol{p}'')$, for all $k \neq j$. This means $\boldsymbol{p}^*$ cannot be an optimal schedule. So, the optimal schedule $\boldsymbol{p}^*$ for $|\boldsymbol{B}|=1$ where $\mathcal{B} \neq \{1\}$ must have zero probability for all states in which node 2 listens, i.e., node 2 always transmits in $\boldsymbol{p}^*$. This implies that if $j \neq 1$,  $R_1(\boldsymbol{p}^*)=0$, and condition~\eqref{eq:algo-check-condition} cannot be satisfied for $k=1$ (contradiction). So, for $|\mathcal{B}|=1$, we only need to check for $\mathcal{B}=\{1\}$.

Next, consider $2 \leq B \leq  K-1$. If an optimal schedule exists, it must satisfies the following conditions:
\begin{align}
&\boldsymbol{p}^* \in \arglmax_{\boldsymbol{p}} \{R_{b_1}(\boldsymbol{p}) = R_{b_2}(\boldsymbol{p}) = \dotsm = R_{b_B} (\boldsymbol{p}) \},\label{eq:condition-1-B}\\
&R_{b_1}(\boldsymbol{p}^*) < R_k(\boldsymbol{p}^*), \forall k \notin \mathcal{B}.\label{eq:condition-2-B}
\end{align}

First, we show that we need to consider only schedules in which node $(b_B+1)$ always listens. Otherwise, we can increase $R_{b_B} (\cdot)$ by reducing the probability of a state $\boldsymbol{s}'$, in which node $(b_B+1)$ transmits, by a small amount and increase the probability of the state $\boldsymbol{s}''$ (which only differs from $\boldsymbol{s}'$ in which node $(b_B+1)$ listens) by the same amount. Doing this will not affect the reception rate of all upstream nodes, i.e., $R_{b_1}(\cdot), R_{b_2}(\cdot), \dotsc, R_{b_{B-1}}(\cdot)$ (since we assume that $b_1 < b_2 < \dotsm < b_B$). The reception rates of other nodes will only change by a small amount while maintaining \eqref{eq:condition-2-B}. So, the new schedule will also achieve the same overall rate but with $|\mathcal{B}|=B-1$, and we would have found this optimal schedule earlier in Step \ref{hd-step-next-B} in Algorithm~\ref{algorithm1}. So, we can ignore the schedule in which node $(b_B+1)$ transmits with non-zero probability.

Having established that node $(b_B+1)$ always listens, we now show that we only need to consider the case where $\mathcal{B}=\{1,2,\dotsc, B\}$. Suppose that $\mathcal{B} \neq \{1,2,\dotsc,B\}$. This means there exists some index $i \leq B$ which is not in the set $\mathcal{B}$ and the last index in the set $\mathcal{B}$ must be strictly greater than $B$, i.e., $b_B > B$. Since $i \notin \mathcal{B}$, from \eqref {eq:condition-2-B}, the reception rate of node $(i+1)$ must be strictly greater than $R^*$, and this also means node $(i+1)$ listens with non-zero probability. Using the same argument as before, we can decrease $R_i(\cdot)$ by a small amount by decreasing the probability of the state $\boldsymbol{s}'$ in which node $(i+1)$ listens by a small amount $\epsilon$ and increasing the probability of state $\boldsymbol{s}''$ (which only differs from  $\boldsymbol{s}'$ in which node $(i+1)$ transmits) by $\epsilon$. Doing this will not affect the reception rate of all upstream nodes, i.e., $\{R_j(\cdot): 1 \leq j \leq j-1\}$. However, the reception rates of all downstream nodes can either increase (if they listen in $\boldsymbol{s}''$) or stay the same (if they transmit in $\boldsymbol{s}''$). In particular, $R_{b_B}(\cdot)$ will increase as node $(b_B+1)$ always listens. This means using the new schedule, we can attain the optimal DF rate with fewer bottleneck nodes. We argued that we would have found this solution earlier in Step~\ref{hd-step-next-B} in Algorithm~\ref{algorithm1}. So, we can ignore this case. Up to this point, we have proven that for $|\mathcal{B}|=B$ in Algorithm~\ref{algorithm1}, we only need to consider $\mathcal{B}=\{1,2,\dotsc,B\}$. 

\subsection{To prove that we only need to consider the global maximum in  \eqref{eq:algo-cond-1}}

Note that in the scheduling problem for the phase-fading HD-MRC, $R_i(\boldsymbol{p})$ is a linear function of $\boldsymbol{p}$ (see \eqref{eq:R-rate}). So, the intersection of linear functions
\begin{equation}
R(\boldsymbol{p}) = R_{b_1}(\boldsymbol{p}) = R_{b_2}(\boldsymbol{p}) = \dotsm = R_{b_B}(\boldsymbol{p})
\end{equation}
is also a linear in $\boldsymbol{p}$. Furthermore, the set of $\{\boldsymbol{p}\}$ that satisfied the aforementioned equation is a convex set. This means we can replace the local maximum function in \eqref{eq:algo-cond-1} by the global maximum function in \eqref{eq:algo2-condition1}.

\subsection{$R(\mathcal{B},\boldsymbol{p}')$ in \eqref{eq:algo-cond-2}, if exists, will not increase with a larger $|\mathcal{B}|$}

Next we will show that once a solution is found for \eqref{eq:algo2-condition1} and \eqref{eq:algo2-condition2}  in Algorithm~\ref{algorithm2}, we can terminate the algorithm.

We will prove this by contradiction. Suppose that there exists a solution for \eqref{eq:algo2-condition1} and \eqref{eq:algo2-condition2} for some $B=B'$, and suppose that we can get a better solution for some $B=B'' > B'$. This is equivalent to saying that there exist some $\boldsymbol{p}',\boldsymbol{p}'' \in \mathcal{G}$ and $B' < B'' \leq K$ such that
\begin{align}
&\boldsymbol{p}' \in \argmax_{\boldsymbol{p}} \{ R_1(\boldsymbol{p}), R_2(\boldsymbol{p}), \dotsc, R_{B'}(\boldsymbol{p}) \}, \label{eq:contradiction-global-max}\\
&R' = R_1(\boldsymbol{p}') < R_j(\boldsymbol{p}'), \quad \forall j > B',\\
&\boldsymbol{p}'' \in \argmax_{\boldsymbol{p}} \{ R_1(\boldsymbol{p}), R_2(\boldsymbol{p}), \dotsc, R_{B'}(\boldsymbol{p}), \dotsc, R_{B''}(\boldsymbol{p}) \},\\
&R'' = R_1(\boldsymbol{p}'') < R_j(\boldsymbol{p}''), \quad \forall j > B'',
\end{align}
where
\begin{equation}
R'< R''.
\end{equation}
Now, since
\begin{subequations}
\begin{align}
R'' &= R_1(\boldsymbol{p}'') = R_2(\boldsymbol{p}'') = \dotsm = R_{B'}(\boldsymbol{p}'') \\ &> R' = R_1(\boldsymbol{p}') = R_2(\boldsymbol{p}') = \dotsm = R_{B'}(\boldsymbol{p}'),
\end{align}
\end{subequations}
eqn.~\eqref{eq:contradiction-global-max} cannot be true as $\boldsymbol{p}'$ cannot be a global maximum for the set $\{1,2,\dotsc,B'\}$ (contradiction).

This completes the proof for Algorithm~\ref{algorithm2}. \hfill $\blacksquare$

\section{The four-node Phase-Fading HD-MRC}\label{appendix-example}

We let
\begin{equation}
p_0=p(L,L), p_1=p(L,T), p_2=p(T,L), p_3=p(T,T).
\end{equation}
From Theorem~\ref{thm:df-gaussian-hd-krc}, the optimal scheduling problem can be written as
\begin{equation}\label{eq:hd-minimax-testing-four-mrc}
\max_{\boldsymbol{p} \in \mathcal{G}} \min \{ r_2(\boldsymbol{p}), r_3(\boldsymbol{p}), r_4(\boldsymbol{p})\},
\end{equation}
where $\boldsymbol{p} = (p_0, p_1, p_2, p_3)$ is a schedule,
$\mathcal{G} = \{ (p_0,p_1,p_2,p_3) \in \mathcal{R}^4: \sum_i p_i \leq 1, p_i \geq 0, \forall i=0,1,2,3\}$ is the set of all feasible schedules, and
\begin{align}
r_2 &= p_0 \Gamma(\lambda_{1,2}P_1/N_2) + p_1 \Gamma(\lambda_{1,2}P_1/N_2)\label{eq:r-2}\\
r_3 &= p_0 \Gamma(\lambda_{1,3}P_1/N_3) + p_2 \Gamma((\lambda_{1,3}P_3+\lambda_{2,3}P_2)/N_3)\label{eq:r-3}\\
r_4 &= p_0 \Gamma(\lambda_{1,4}P_1/N_4) + p_1 \Gamma((\lambda_{1,4}P_1+\lambda_{3,4}P_3)/N_4)\nonumber \\ &\quad  + p_2 \Gamma((\lambda_{1,4}P_1 + \lambda_{2,4}P_2)/N_4)\nonumber\\
&\quad + p_3 \Gamma((\lambda_{1,4}P_1 + \lambda_{2,4}P_2 + \lambda_{3,4}P_3)/N_4).\label{eq:r-4}
\end{align}

Using Algorithm~\ref{algorithm2}, we start with $\mathcal{B} = \{1\}$. We need to find some $\boldsymbol{p}^*$ such that
\begin{align}
&\boldsymbol{p}^* \in \max_{\boldsymbol{p} \in \mathcal{G}} r_2(\boldsymbol{p}),\\
&r_2(\boldsymbol{p}^*) < r_3(\boldsymbol{p}^*), \\
&r_2(\boldsymbol{p}^*) < r_4(\boldsymbol{p}^*).
\end{align}
The first condition means that $p_0^* + p_1^* = 1$, and $p_2^*=p_3^*=0$. The second and the third conditions simplify to
\begin{align}
&\frac{\Gamma(\lambda_{1,2}P_1/N_2)}{\Gamma(\lambda_{1,3}P_1/N_3)} < p_0^* \leq 1, \label{eq:hd-case-3-4mrc-a} \\
&0 \leq p_0^* < \frac{\Gamma((\lambda_{1,4}P_1+\lambda_{3,4}P_3)/N_4)-\Gamma(\lambda_{1,2}P_1/N_2)}{\Gamma((\lambda_{1,4}P_1+\lambda_{3,4}P_3)/N_4)-\Gamma(\lambda_{1,4}P_1/N_4)}, \label{eq:hd-case-3-4mrc-b}
\end{align}
respectively.
If we can find some $p_0^* \in [0,1]$ satisfying \eqref{eq:hd-case-3-4mrc-a} and \eqref{eq:hd-case-3-4mrc-b}, then the optimal schedule is given by $\boldsymbol{p}^* = ( p_0^*, 1-p_0^*, 0, 0)$.

If conditions \eqref{eq:hd-case-3-4mrc-a} and \eqref{eq:hd-case-3-4mrc-b} cannot be satisfied, we proceed to $\mathcal{B} = \{1,2\}$. We need to find some $\boldsymbol{p}^*$ such that
\begin{align}
\boldsymbol{p}^* &\in \max_{\boldsymbol{p} \in \mathcal{G}} \{r_2(\boldsymbol{p})=r_3(\boldsymbol{p})\},\\
r_2(\boldsymbol{p}^*) &=r_3(\boldsymbol{p}^*) < r_4(\boldsymbol{p}^*).\label{eq:case-r2=r3}
\end{align}
Maximizing $r_2(\boldsymbol{p})$ and $r_3(\boldsymbol{p})$ gives $p_0^*+p_2^* = 1$ and $p_1^*=p_3^*=0$. Furthermore,  $r_2(\boldsymbol{p}^*) = r_3(\boldsymbol{p}^*)$ gives $p_0^*\Gamma\left(\frac{\lambda_{1,2}P_1}{N_2}\right) = p_0^* \Gamma\left(\frac{\lambda_{1,3}P_1}{N_3}\right) + p_2^* \Gamma\left(\frac{\lambda_{1,3}P_3+\lambda_{2,3}P_2}{N_3}\right)$. So, we have
\begin{align}
p_0^* &=
\frac{\Gamma\left(\frac{\lambda_{1,3}P_3+\lambda_{2,3}P_2}{N_3}\right)}{\Gamma\left(\frac{\lambda_{1,3}P_3+\lambda_{2,3}P_2}{N_3}\right)+\Gamma\left(\frac{\lambda_{1,2}P_1}{N_2}\right)-\Gamma\left(\frac{\lambda_{1,3}P_1}{N_3}\right)}\label{eq:hd-case-6-1}\\
p_2^* &=
\frac{\Gamma\left(\frac{\lambda_{1,2}P_1}{N_2}\right)-\Gamma\left(\frac{\lambda_{1,3}P_1}{N_3}\right)}{\Gamma\left(\frac{\lambda_{1,3}P_3+\lambda_{2,3}P_2}{N_3}\right)+\Gamma\left(\frac{\lambda_{1,2}P_1}{N_2}\right)-\Gamma\left(\frac{\lambda_{1,3}P_1}{N_3}\right)}. \label{eq:hd-case-6-2}
\end{align}
The second condition \eqref{eq:case-r2=r3} is thus equivalent to
\begin{multline}
\Gamma\left(\frac{\lambda_{1,3}P_3+\lambda_{2,3}P_2}{N_3}\right)\left[ \Gamma\left(\frac{\lambda_{1,2}P_1}{N_2}\right)-\Gamma\left(\frac{\lambda_{1,4}P_1}{N_4}\right) \right]\\ < \Gamma\left(\frac{\lambda_{1,4}P_4+\lambda_{2,4}P_2}{N_4}\right) \left[\Gamma\left(\frac{\lambda_{1,2}P_1}{N_2}\right)-\Gamma\left(\frac{\lambda_{1,3}P_1}{N_3}\right)\right].\label{eq:hd-case-6-3}
\end{multline}

If condition \eqref{eq:hd-case-6-3} cannot be satisfied, we proceed to $\mathcal{B}=\{1,2,3\}$. We need to find some $\boldsymbol{p}^*$ such that
\begin{equation}
\boldsymbol{p}^* \in \max_{\boldsymbol{p} \in \mathcal{G}} \{r_2(\boldsymbol{p}) = r_3(\boldsymbol{p}) = r_4(\boldsymbol{p})\},
\end{equation}
We have the following linearly independent equations:
\begin{align}
r_2(\boldsymbol{p}^*) &= r_3(\boldsymbol{p}^*)\\
r_2(\boldsymbol{p}^*) &= r_4(\boldsymbol{p}^*)\\
p_0^* + p_1^* + p_2^* + p_3^* &= 1.
\end{align}
Solving the above, we can express $p_1$, $p_2$, $p_3$ in terms of only $p_0$. Hence, the optimization becomes
\begin{equation}\label{eq:hd-4mrc-case-7}
\max r_2(p_0),
\end{equation}
subject to $0 \leq p_i \leq 1$, for $i=0,1,2,3$.

Summarizing all the above, we have the following for the four-node phase-fading HD-MRC:\\
$\bullet$ If
$$\frac{\Gamma(\lambda_{1,2}P_1/N_2)}{\Gamma(\lambda_{1,3}P_1/N_3)}< \frac{\Gamma((\lambda_{1,4}P_1+\lambda_{3,4}P_3)/N_4)-\Gamma(\lambda_{1,2}P_1/N_2)}{\Gamma((\lambda_{1,4}P_1+\lambda_{3,4}P_3)/N_4)-\Gamma(\lambda_{1,4}P_1/N_4)},$$ $$ 0 \leq \frac{\Gamma((\lambda_{1,4}P_1+\lambda_{3,4}P_3)/N_4)-\Gamma(\lambda_{1,2}P_1/N_2)}{\Gamma((\lambda_{1,4}P_1+\lambda_{3,4}P_3)/N_4)-\Gamma(\lambda_{1,4}P_1/N_4)},$$ and $$ \frac{\Gamma(\lambda_{1,2}P_1/N_2)}{\Gamma(\lambda_{1,3}P_1/N_3)} \leq 1,$$
then an optimal schedule is given by some $p_0^*$ satisfying \eqref{eq:hd-case-3-4mrc-a} and \eqref{eq:hd-case-3-4mrc-b}, $p_1^*=1-p_0^*$, $p_2^*=0$, and $p_3^*=0$.\\
$\bullet$ Else, if
$$
0 < \frac{\Gamma\left(\frac{\lambda_{1,3}P_3+\lambda_{2,3}P_2}{N_3}\right)}{\Gamma\left(\frac{\lambda_{1,3}P_3+\lambda_{2,3}P_2}{N_3}\right)+\Gamma\left(\frac{\lambda_{1,2}P_1}{N_2}\right)-\Gamma\left(\frac{\lambda_{1,3}P_1}{N_3}\right)} < 1
$$
and \eqref{eq:hd-case-6-3} are both satisfied, an optimal schedule is given by \eqref{eq:hd-case-6-1}, \eqref{eq:hd-case-6-2}, $p_1^*=0$, and $p_3^*=0$.\\
$\bullet$ Else, we solve for \eqref{eq:hd-4mrc-case-7}. Note that this is simpler than the original optimization problem \eqref{eq:hd-minimax-testing-four-mrc}.



\begin{IEEEbiographynophoto}{Lawrence Ong}
(S'05--M'10) received the BEng (1st Hons) degree in electrical engineering from the National University of Singapore (NUS), Singapore, in 2001. He subsequently received the MPhil degree from the University of Cambridge, UK, in 2004 and the PhD degree from NUS in 2008.

He was with MobileOne, Singapore, as a system engineer from 2001 to 2002. He was a research fellow at NUS, from 2007 to 2008. From 2008 to 2012, he was a postdoctoral researcher at The University of Newcastle, Australia.

In 2012, he was awarded the Discovery Early Career Researcher Award (DECRA) by the Australian Research Council (ARC). He is currently a DECRA fellow at The University of Newcastle.
\end{IEEEbiographynophoto}

\begin{IEEEbiographynophoto}{Mehul Motani}
(S'92--M'00) received the B.S. degree from Cooper Union, New York, NY, the M.S. degree from Syracuse University, Syracuse, NY, and the Ph.D. degree from Cornell University, Ithaca, NY, all in Electrical and Computer Engineering. 

Dr. Motani is currently an Associate Professor in the Electrical and Computer Engineering Department at the National University of Singapore (NUS). He has held a Visiting Fellow appointment at Princeton University, Princeton, NJ. Previously, he was a Research Scientist at the Institute for Infocomm Research in Singapore for three years and a Systems Engineer at Lockheed Martin in Syracuse, NY for over four years. His research interests are in the area of wireless networks. Recently he has been working on research problems which sit at the boundary of information theory, networking, and communications, with applications to mobile computing, underwater communications, sustainable development and societal networks. 

Dr. Motani has received the Intel Foundation Fellowship for his Ph.D. research, the NUS Faculty of Engineering Innovative Teaching Award, and placement on the NUS Faculty of Engineering Teaching Honours List. He has served on the organizing committees of ISIT, WiNC and ICCS, and the technical program committees of MobiCom, Infocom, ICNP, SECON, and several other conferences. He participates actively in IEEE and ACM and has served as the secretary of the IEEE Information Theory Society Board of Governors. He is currently an Associate Editor for the IEEE Transactions on Information Theory and an Editor for the IEEE Transactions on Communications. 
\end{IEEEbiographynophoto}

\begin{IEEEbiographynophoto}{Sarah J. Johnson}
(M'04) received the B.E. (Hons) degree in electrical engineering in 2000, and PhD in 2004, both from the University of Newcastle, Australia.

She has held research positions with the Wireless Signal Processing
Program, National ICT Australia and the University of Newcastle. In
2011, Dr. Johnson was awarded a Future Fellowship by the Australian
Research Council (ARC). She is currently a future fellow at The
University of Newcastle.

\end{IEEEbiographynophoto}

\end{document}